\documentclass[review]{elsarticle}
\usepackage{lineno,hyperref}
\usepackage{amsopn}
\usepackage{amsmath}
\usepackage{amsfonts}
\usepackage{amsthm}
\usepackage{color}
\usepackage{algorithmic}
\usepackage[linesnumbered,ruled,vlined]{algorithm2e}
\newtheorem{definition}{Definition}
\newtheorem{theorem}{Theorem}
\newtheorem{lemma}{Lemma}
\usepackage[title]{appendix}

\usepackage{xparse}%
\usepackage{blindtext}
\usepackage{morewrites}
\usepackage{wrapfig}
\usepackage{xkeyval}%
\newwrite\authorbibfile%
\AtBeginDocument{%
  \immediate\openout\authorbibfile=\jobname.aub%
}%
\AtEndDocument{%
\immediate\closeout\authorbibfile
\InputIfFileExists{\jobname.aub}{}{}
}%

\makeatletter

\define@key{authorbib}{scale}[1]{%
\def\AuthorbibKVMacroScale{#1}%
}

\define@key{authorbib}{wraplines}[10]{%
\def\AuthorbibKVMacroWraplines{#1}%
}

\define@key{authorbib}{imagewidth}[4cm]{%
\def\AuthorbibKVMacroImagewidth{#1}%
}
\define@key{authorbib}{overhang}[10pt]{%
\def\AuthorbibKVMacroOverhang{#1}%
}

\define@key{authorbib}{imagepos}[l]{%
\def\AuthorbibKVMacroImagepos{#1}%
}

\makeatother

\presetkeys{authorbib}{imagepos=l,imagewidth=4cm,wraplines=15,overhang=20pt}{}

\newlength{\AuthorbibTopSkip}
\newlength{\AuthorbibBottomSkip}
\NewDocumentCommand{\authorbibliography}{+o+m+m+m}{%
  \IfNoValueTF{#1}{%
  }{%
    \setkeys{authorbib}{#1}%
    \immediate\write\authorbibfile{%
      \string\begin{wrapfigure}[\AuthorbibKVMacroWraplines]{\AuthorbibKVMacroImagepos}[\AuthorbibKVMacroOverhang]{\AuthorbibKVMacroImagewidth}^^J
        \string\includegraphics[scale=\AuthorbibKVMacroScale]{#2}^^J
        \string\end{wrapfigure}^^J
    }%
  }%
  \IfNoValueTF{#3}{%
    \typeout{Warning: No author name}%
  }{%
      \immediate\write\authorbibfile{%
      \unexpanded{\vspace{\AuthorbibTopSkip}}^^J
      \string\noindent\relax
      \unexpanded{\textbf{#3}\par}^^J
      \string\noindent\relax
      \unexpanded{#4}^^J%
      \unexpanded{\vspace{\AuthorbibBottomSkip}}^^J
      }%
  }%
}%

\journal{Journal of Systems Architecture}
\biboptions{sort&compress}

\begin{document}

\begin{frontmatter}

\title{Joint Task Offloading and Resource Optimization in NOMA-Based Vehicular Edge Computing: A Game-Theoretic DRL Approach}

\author[cqu]{Xincao Xu}
\ead{near@cqu.edu.cn}

\author[cqu,nf]{Kai Liu\corref{mycorrespondingauthor}}
\cortext[mycorrespondingauthor]{Corresponding author}
\ead{liukai0807@cqu.edu.cn}

\author[swjtu]{Penglin Dai}
\ead{penglindai@swjtu.edu.cn}

\author[cqu]{Feiyu Jin}
\ead{fyjin@cqu.edu.cn}

\author[cqu]{Hualing Ren}
\ead{renharlin@cqu.edu.cn}

\author[scnu]{Choujun Zhan}
\ead{zchoujun2@gmail.com}

\author[cqu]{Songtao Guo}
\ead{guosongtao@cqu.edu.cn}

\address[cqu]{College of Computer Science, Chongqing University, Chongqing, China}
\address[nf]{School of Electrical and Computer Engineering, Nanfang College, Guangzhou, China}
\address[swjtu]{School of Computing and Artificial Intelligence, Southwest Jiaotong University, Chengdu, China}
\address[scnu]{School of Computing, South China Normal University, Guangzhou, China}

\begin{abstract}
Vehicular edge computing (VEC) becomes a promising paradigm for the development of emerging intelligent transportation systems. 
Nevertheless, the limited resources and massive transmission demands bring great challenges on implementing vehicular applications with stringent deadline requirements. 
This work presents a non-orthogonal multiple access (NOMA) based architecture in VEC, where heterogeneous edge nodes are cooperated for real-time task processing. 
We derive a vehicle-to-infrastructure (V2I) transmission model by considering both intra-edge and inter-edge interferences and formulate a cooperative resource optimization (CRO) problem by jointly optimizing the task offloading and resource allocation, aiming at maximizing the service ratio. 
Further, we decompose the CRO into two subproblems, namely, task offloading and resource allocation.
In particular, the task offloading subproblem is modeled as an exact potential game (EPG), and a multi-agent distributed distributional deep deterministic policy gradient (MAD4PG) is proposed to achieve the Nash equilibrium.
The resource allocation subproblem is divided into two independent convex optimization problems, and an optimal solution is proposed by using a gradient-based iterative method and KKT condition.
Finally, we build the simulation model based on real-world vehicular trajectories and give a comprehensive performance evaluation, which conclusively demonstrates the superiority of the proposed solutions.
\end{abstract}


\end{frontmatter}


\section{Introduction}
Recent advances in vehicular networks have paved the way for the development of emerging intelligent transportation systems (ITSs) such as cooperative autonomous driving \cite{bagheri20215g} and vehicular cyber-physical systems \cite{xu2022age}. 
Nevertheless, it demands massive data transmission and intensive task computation to enable most of the applications. 
For instance, modern vehicles such as Tesla Model X have already equipped with eight cameras, 12 ultrasonic radars, and one millimeter-wave radar, and the data computation requirements are ever-increasing. On the other hand, the limited communication and computation resources in vehicular networks make it non-trivial to support real-time vehicular applications.  Clearly, it is imperative to investigate efficient real-time task offloading and heterogeneous resource allocation in vehicular networks.

Vehicular edge computing (VEC) \cite{liu2019hierarchical} has recently emerged as a promising paradigm to facilitate task processing at the edge of vehicular networks. 
Great efforts have been paid to the development of VEC \cite{liu2020fog, dai2021edge, zhang2021digital, liu2020adaptive, liu2017coding, xu2020vehicular}, where the edge node such as 5G base stations and roadside units (RSUs) with collocated computation units, can process tasks with the data uploaded by vehicles via vehicle-to-infrastructure (V2I) communications. 
However, none of them have exploited the non-orthogonal multiple access (NOMA) \cite{islam2017power} technology to enhance the network capacity. 
Some studies have incorporated NOMA in vehicular communications \cite{patel2021performance, zhang2021centralized, zhu2021decentralized, liu2019energy}, where the vehicles may communicate with the edge node over the same frequency of bandwidth with different transmission power. 
However, these studies only considered single edge node scenario and cannot deal with the interference among different edge nodes. 
To improve system reliability, a few studies have taken resource allocation into consideration to counteract the effect of time-varying V2I channel conditions and dynamic available computation resources in VEC \cite{liu2021rtds, liu2022a, chen2020robust, liu2014temporal, xu2021potential, liu2015cooperative}. 
Nevertheless, none of them investigated the synergistic effects of real-time task offloading and communication/computation resource allocation.
Some studies designed allocation mechanisms for both communication and computation resources to improve resource efficiency \cite{cui2021reinforcement, han2020reliability, xu2021socially}.
A few literatures formulated the joint optimization model by integrating task offloading and resource allocation \cite{dai2021asynchronous, dai2020probabilistic}.
However, existing studies are mainly based on centralized scheduling, which may hurdle system scalability in vehicular networks.
The multi-agent deep reinforcement learning (MADRL) \cite{althamary2019survey} as an emerging distributed solution has been proposed for vehicular applications \cite{alam2022multi, zhang2021adaptive, nie2021semi}.
On the other hand, some of the work combined reinforcement learning and game theory \cite{zheng2022stackelberg, albaba2021driver, rajeswaran2020game} to solve complex optimization problems.
However, none of the solutions can be directly applied in vehicular networks for joint real-time task offloading and heterogeneous resource allocation.

With above motivation, this paper investigates a distributed scheduling solution for joint task offloading and resource allocation based on the multi-agent distributed distributional deep deterministic policy gradient (MAD4PG) and potential game theory.
In particular, we first model the task offloading decision-making process as an exact potential game (EPG) \cite{chew2016potential} with Nash equilibrium (NE) existence and convergence under a designed potential function, where the edge nodes are rational players to maximize their profits (i.e., the service ratio of real-time tasks, which is the number of tasks completed before their deadlines over the number of total tasks).
According to potential game theory, the NE can be achieved by maximizing the potential of each edge node with the designed potential function.
Naturally, the potential function is well-suitable as the rewards for edge nodes in the proposed MAD4PG algorithm.
The remaining resource allocation problem is divided into two independent convex optimization problems, and an optimal solution is proposed based on a gradient-based iterative method and Karush-Kuhn-Tucher (KKT) condition \cite{boyd2004convex}.

This work puts the first effort on jointly investigating the real-time task offloading and heterogeneous resource allocation in NOMA-based VEC by addressing the following challenges. 
First, the V2I uplinks suffer interference from vehicles using the same channel, and the influence depends on the transmission power allocated by edge nodes. 
Second, there is a serious imbalance of workload distribution among different edge nodes due to the time-varying distribution of computation-intensive and delay-sensitive tasks.
Third, it is non-trivial to make edge nodes determine task offloading and resource allocation decisions independently and efficiently only with their local knowledge.
Thus, it is imperative yet challenging to study an effective and distributed method for joint real-time task offloading and heterogeneous resource allocation in NOMA-based VEC.
The main contributions of this work are outlined as follows.

\begin{itemize}
	\item We present a NOMA-based VEC architecture, where the vehicles share the same frequency of bandwidth resources and communicate with the edge node with the allocated transmission power. The tasks arrive stochastically at vehicles and have different computation resource requirements and deadlines, which can be further uploaded to the edge nodes for computing via V2I communications. The edge nodes with heterogeneous computation capabilities, i.e., CPU clock frequencies, can either execute the tasks locally with allocated computation resources or migrate the tasks to neighboring edge nodes via wired connections.
	\item We formulate a cooperative resource optimization (CRO) problem, which jointly offloads tasks and allocates communication and computation resources to maximize the service ratio. Specifically, we derive a V2I transmission model, in which both intra-edge and inter-edge interferences are modeled based on the NOMA principle. Then, we derive a task offloading model by considering the cooperation of heterogeneous edge nodes. 
	\item We decompose the CRO into two subproblems, namely, task offloading and resource allocation. Specifically, we model the first subproblem as a non-cooperative game among edge nodes, which is further proved as an EPG with NE existence and convergence. Then, we design a MAD4PG algorithm, which is a multi-agent version of D4PG \cite{barth2018distributed}, to achieve the NE, where edge nodes act as independent agents to determine the task offloading decisions by adopting the achieved potential as rewards. Further, we divide the second subproblem into two independent convex problems and derive an optimal solution based on the gradient-based iterative method and KKT condition.
	\item We build the simulation model based on real-world vehicular trajectories. Then, in addition to the primary metrics, cumulative reward and average service ratio, we design another four metrics, including average processing time, average service time, average achieved potential, and proportion of locally processing to migration, to give insights into performance evaluation. Further, we implement the proposed algorithms, as well as four competitive solutions, i.e., optimal resource allocation and task migration only (ORM), optimal resource allocation and task local processing only (ORL), distributed distributional deep deterministic policy gradient (D4PG) \cite{barth2018distributed} for joint task offloading and resource allocation, and multi-agent deep deterministic policy gradient (MADDPG) \cite{zhang2021adaptive} for task offloading based on optimal resource allocation. The simulation results conclusively demonstrate the superiority of the proposed algorithms.
\end{itemize}

The rest of this paper is organized as follows.
Section II reviews the related work.
Section III presents the system architecture.
Section IV formulates the CRO problem.
Section V proposes the solutions.
Section VI evaluates the performance.
Finally, Section VII concludes this paper and discusses future research directions.

\section{Related Work}

Great efforts have been devoted to vehicular applications at the edge of vehicular networks. 
Liu et al. \cite{liu2020fog} investigated the cooperative data dissemination problem in an end-edge-cloud cooperation architecture. A clique-based algorithm was proposed to schedule data encoding and dissemination jointly. 
Dai et al. \cite{dai2021edge} designed a VEC architecture for adaptive-bitrate-based multimedia streaming, where edge cache and transmission services are provided for file chunks encoded with different quality levels. 
Zhang et al. \cite{zhang2021digital} presented a socially aware vehicular edge caching technique that dynamically orchestrates the cache capabilities of edge nodes and intelligent vehicles based on user preference similarity and service availability.
Liu et al. \cite{liu2020adaptive} presented a two-layer VEC architecture to exploit the cloud, static edge nodes, and mobile edge nodes for processing time-critical tasks. 
Liu et al. \cite{liu2017coding} proposed a memetic algorithm to exploit the synergistic effects between vehicular caching and network coding for enhancing the bandwidth efficiency of data broadcasting in VEC. 
Xu et al. \cite{xu2020vehicular} proposed a vehicle collision warning strategy based on trajectory calibration and considering V2I communication delay and packet loss. 
However, the existing VEC architectures of these studies are based on the conventional orthogonal multiple access (OMA), and none of them have exploited the NOMA technology to enhance the network capacity.

Several researchers have exploited the NOMA technology in vehicular networks to further improve bandwidth efficiency.
Patel et al. \cite{patel2021performance} evaluated the communication capacity of NOMA-based vehicular networks, and the numerical results show that NOMA outperforms the conventional OMA by approximately 20\%. 
Zhang et al. \cite{zhang2021centralized} developed a centralized two-stage resource allocation strategy for NOMA-integrated vehicular networks using a graph-based matching approach and distributed power control via a non-cooperative game. 
Zhu et al. \cite{zhu2021decentralized} proposed an optimal power allocation strategy that considers stochastic task arrival and channel fluctuation to maximize the long-term power consumption and delay. 
Liu et al. \cite{liu2019energy} proposed an alternate direction algorithm for the multipliers technique to allocate power in NOMA-based autonomous driving vehicle networks.
Nevertheless, these studies are mainly based on single edge node scenario, while the interference among different edge nodes cannot be handled. 

Some studies have focused on task offloading or resource allocation in VEC.
Liu et al. \cite{liu2021rtds} proposed a real-time distributed method for multi-period task offloading by evaluating the mobility-aware communication model, the resources-aware computation model, and the deadline-aware award model in VEC.
Liu et al. \cite{liu2022a} presented an algorithm that combined the alternating direction method of multipliers and particle swarm optimization for task offloading to minimize the weighted sum of execution delay, energy consumption, and payment cost.
Chen et al. \cite{chen2020robust} presented a computation offloading approach with failure recovery to reduce energy usage and shorten application completion time. 
Liu et al. \cite{liu2014temporal} proposed a heuristic scheduling algorithm for real-time data dissemination by considering both the time constraint of data dissemination and the freshness of data items.
Xu et al. \cite{xu2021potential} proposed an incentive-based probability update and strategy selection algorithm for subchannel allocation by modeling the channel allocation problem as a potential game.
Liu et al. \cite{liu2015cooperative} developed a greedy method for cooperative data dissemination in a hybrid vehicular communication environment.
Nevertheless, none of them investigated the synergistic effects of real-time task offloading and communication/computation resource allocation.

Several studies have considered joint communication and computation resource allocations in VEC. 
Cui et al. \cite{cui2021reinforcement} proposed a multi-objective reinforcement learning method to reduce system delay by combining communication and computation resource allocation.
Han et al. \cite{han2020reliability} presented the coupling-oriented reliability calculation for cooperative vehicular communication and computing via dynamic programming methods. 
Xu et al. \cite{xu2021socially} employed contract theory to allocate communication and computation resources for each prospective content supplier and content requester pair.
A few researchers have studied joint task offloading and resource allocation.
Dai et al. \cite{dai2021asynchronous} proposed a asynchronous deep reinforcement learning for data-driven task offloading considering heterogeneous servers, such as powerful vehicles, VEC servers, and the cloud.
Dai et al. \cite{dai2020probabilistic} developed a probabilistic computation offloading approach for independently scheduling computation offloading based on the calculated allocation probability in edge nodes. 
However, the existing studies cannot be applied to large-scale vehicle networks, as they are mainly based on centralized scheduling with high communication overhead and scheduling complexity.

Some studies have focused on task offloading or resource allocation by using MADRL \cite{althamary2019survey} in vehicular networks.
Alam et al. \cite{alam2022multi} developed a multi-agent DRL-based Hungarian algorithm (MADRLHA) for dynamic task offloading in VEC to guarantee latency, energy consumption, and payment cost requirements.
Zhang et al. \cite{zhang2021adaptive} presented a MADDPG approach for edge resource allocation to minimize vehicular task offloading cost under strict delay constraints.
Nie et al. \cite{nie2021semi} proposed a multi-agent federated reinforcement learning (MAFRL) algorithm to optimize resource allocation, user association, and power control jointly in an unmanned aerial vehicle (UAV)-enabled VEC.
The combination of game theory and reinforcement learning has recently attracted much academic attention.
Zhang et al. \cite{zheng2022stackelberg} modeled the interaction between the actor and critic of actor-critic-based reinforcement learning as a two-play general-sum game with a leader-follower structure.
Albaba et al. \cite{albaba2021driver} combined the DQN and hierarchical game theory for behavioral predictions of drivers in highway driving scenarios, where level-k reasoning is used to model the decision-making process of human drivers.
Rajeswaran et al. \cite{rajeswaran2020game} developed a framework that casts model-based reinforcement learning as a Stackelberg game between a policy player and a model player.
However, none of the solutions can be directly applied in vehicular networks for joint real-time task offloading and heterogeneous resource allocation.

\section{System Architecture}

In this section, we propose a NOMA-based architecture for cooperative communication and computation in VEC. 
As illustrated in Fig. \ref{fig_1}, the infrastructures (e.g., $e_1$$\sim$$e_3$) installed along the roadside such as 5G base stations and RSUs, have different equipped computing units (i.e., CPU chips), which are considered as edge nodes to accelerate computation task wireless offloaded by mobile vehicles.
The tasks arrive stochastically at vehicles, which may contain different data to be computed. 
The vehicle can communicate with the edge node via V2I communications within its communication range.
Then, the vehicles upload the task to a nearby edge node, where the transmission power is allocated by the corresponding edge node.
By employing superposition coding (SC) at vehicles and successive interference cancellation (SIC) at edge nodes \cite{khan2021noma}, vehicles can share the same frequency of bandwidth resources. 
In particular, the signals of strong vehicles are decoded and canceled by the edge nodes in succession before decoding the signals of weak vehicles in NOMA-based VEC.
In addition, the edge nodes are connected via a wired network.
Further, the edge nodes decide whether to execute the received task locally or migrate it to other nodes via wired connections.
Finally, the edge nodes allocate the computation resources for processing tasks.

\begin{figure}
\centering
  \includegraphics[width=0.75\columnwidth]{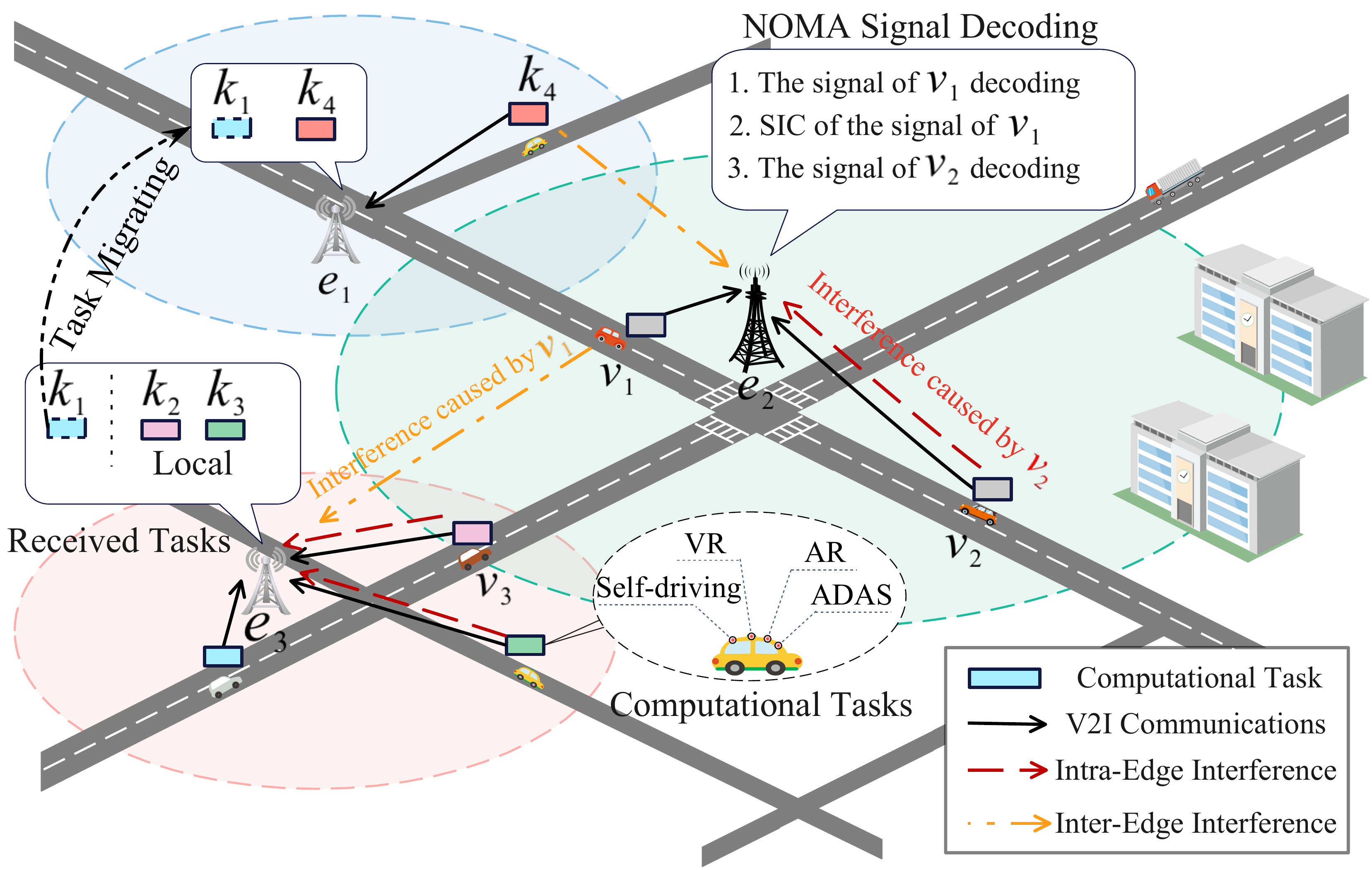}
  \caption{NOMA-based vehicular edge computing architecture}
  \label{fig_1}
\end{figure} 

The system characteristics are summarized as follows.
First, the computational tasks requested by vehicles may have different data sizes, computation resource requirements, and deadlines. 
Therefore, the service condition of tasks (i.e., whether the task will be successfully serviced before its deadline) may differ when offloaded to different edge nodes with diverse computation capabilities, i.e., the CPU clock frequencies. 
Second, increasing the transmission power of a vehicle may improve the achieved V2I transmission rate, but can also damage other V2I uplinks due to enhanced intra-edge and inter-edge interferences. 
Moreover, the power allocations of edge nodes vary over time and are unknown to each other. 
Thus, the edge node must determine the transmission power of vehicles by considering the effects of other edge nodes' power allocation. 
Third, the workloads of edge nodes may not be balanced due to the stochastic arrival of tasks and the time-varying distribution of vehicles. 
When an edge node is overburdened, it becomes appropriate to migrate the extra tasks to other edge nodes with excess computation resources to expedite processing. 
However, transmitting the task data to the edge node via wired connections lengthens the task service time. 

Further, we provide an example to demonstrate the concept better. 
As shown in Fig. \ref{fig_1}, vehicles $v_1$ and $v_2$ upload their tasks via V2I communications.
Since the V2I channel condition between edge node $e_2$ and vehicle $v_1$ is better than that of vehicles $v_2$ and $v_3$, the signal of vehicle $v_1$ can be decoded firstly by treating other signals as the noise. 
Then, the signal of vehicle $v_1$ can be canceled by the edge node $e_2$ when decoding the signals of vehicles $v_2$ and $v_3$. 
However, the signal of vehicle $v_1$ may be interfered by vehicle $v_2$ during the V2I transmission; such interference is termed ``intra-edge interference", since vehicles $v_1$ and $v_2$ are within the radio coverage of the same edge node $e_2$.  
On the other hand, the signal of vehicle $v_3$ may be interfered by vehicle $v_1$, and such interference from other edge nodes is termed ``inter-edge interference."
Moreover, it is obvious that the task workload among edge nodes $e_1$ and $e_3$ is uneven, since there are three tasks (i.e., $k_1$, $k_2$, and $k_3$) in the edge node $e_3$ but only one task $k_4$ in the edge node $e_1$. 
Assume that the computation resources of edge node $e_1$ are significantly more than those of edge node $e_3$. 
The task $k_1$ should be migrated to the edge node $e_1$ over the wired network so that it can be serviced in a shorter time.
As demonstrated above, it is crucial yet difficult to take advantage of cooperative communication and computation among edge nodes by designing an effective and distributed scheduling mechanism for real-time task offloading and heterogeneous resource allocation to optimize the overall system performance.

\section{Problem Formulation}

\subsection{Preliminary}
The set of discrete time slots is denoted by $\mathbf{T}=\{1, \ldots, t, \ldots, T\}$, where $T$ is the number of time slots.
The set of vehicles is denoted by $\mathbf{V}=\{1, \ldots, v, \ldots, V\}$, and the location of vehicle $v \in \mathbf{V}$ at time $t$ is denoted by $l_{v}^{t}$.
The task arrival probability of vehicle $v$ at time $t$ is denoted by $\tau_{v}^{t}$, and we denote the set of tasks requested by the vehicle $v$ as $\mathbf{K}_{v}$.
The task requested by vehicle $v$ at time $t$ denoted by $k_{v}^{t} \in \mathbf{K}_{v}$ is characterized by a three-tuple $k_{v}^{t}=\left(d_{k}, c_{k}, t_{k}\right)$, where $d_{k}$, $c_{k}$, and  $t_{k}$ are the data size, CPU cycles for processing one bit of data, and deadline, respectively.
The set of edge nodes is denoted by $\mathbf{E}=\{1, \ldots, e, \ldots, E\}$, and the edge node $e \in \mathbf{E}$ is characterized by a four-tuple $e=\left(p_{e}, c_{e}, u_{e}, l_{e}\right)$, where $p_{e}$ is the maximum power of V2I communications, $c_{e}$ is the computing frequency, $u_e$ is the V2I communication range, and $l_{e}$ is the location.
The transmission rate of wired communication between edge nodes is denoted by $z$.
The distance between vehicle $v$ and edge node $e$ at time $t$ is denoted by $\operatorname{dist}_{v, e}^{t}$.
The set of vehicles within the radio coverage of edge node $e$ at time $t$ is denoted by $\mathbf{V}_{e}^{t}=\left\{v \mid \operatorname{dist}_{v, e}^{t} \leq u_{e}, \forall v \in \mathbf{V}\right\}, \mathbf{V}_{e}^{t} \subseteq \mathbf{V}$.
The bandwidth of V2I communications is denoted by $b$.
The primary notations are summarized in Table \ref{table_notations}.

\begin{table*}[ht]\scriptsize
\centering
\caption{Summary of primary notations}
\begin{tabular}[t]{lll}
\hline
\hline
Notations&Descriptions\\
\hline
$\mathbf{T}$&Set of discrete time slots $\mathbf{T}=\{1, \ldots, t, \ldots, T\}$\\
$\mathbf{V}$&Set of vehicles $\mathbf{V}=\{1, \ldots, v, \ldots, V\}$\\
$\mathbf{E}$&Set of edge nodes, $e \in \mathbf{E}$ and $e=\left(p_{e}, c_{e}, r_{e}, l_{e}\right)$\\
$\tau_{v}^{t}$&Task arrival probability of vehicle $v$ at time $t$\\
$\mathbf{K}_{v}$&Set of computation tasks requested by vehicle $v$\\
$k_{v}^{t}$&Task requested by vehicle $v$ at time $t$, $k_{v}^{t} \in \mathbf{K}_{v}$ and $k_{v}^{t}=\left(d_{k}, c_{k}, t_{k}\right)$\\
$d_k$&Data size of task $k_{v}^{t}$\\
$c_k$&CPU cycles for processing one bit data of task $k_{v}^{t}$\\
$t_{k}$&Deadline of task $k_{v}^{t}$\\
$l_v^t$&Location of vehicle $v$ at time $t$\\
$p_e$&Maximum power of V2I communications at edge node $e$\\
$c_e$&Computing clock frequency of edge node $e$\\
$u_e$&V2I communication range of edge node $e$\\
$l_e$&Location of edge node $e$\\
$\operatorname{dist}_{v, e}^{t}$&Distance between vehicle $v$ and edge node $e$ at time $t$\\
$z$&Wired transmission rate between edge nodes\\
$b$&Bandwidth of V2I communications\\
$p_{v, e}^{t}$&Transmission power of vehicle $v$ allocated by edge node $e$ at time $t$\\
$q_{v, e}^t$&Binary indicates whether task $k_v^t$ is offloaded to edge node $e$\\
$c_{v, e}^t$&Computation resource allocated by edge node $e$ for task of vehicle $v$ at time $t$\\
$\mathbf{V}_{e}^{t}$&Set of vehicles within the coverage of edge node $e$ at time $t$\\
$\mathbf{V}_{h_{v, e}}^{t}$&Set of vehicles that have a worse channel condition than vehicle $v$ at time $t$\\
$\mathbf{K}_{e}^{t}$&Set of tasks uploaded by vehicles that are within the coverage of edge node $e$ at time $t$\\
$\mathbf{K}_{q_e}^{t}$&Set of tasks which are offloaded in edge node $e$ at time $t$\\
\hline
\hline
\end{tabular}
\label{table_notations}
\end{table*}

\begin{figure}
\centering
  \includegraphics[width=1\columnwidth]{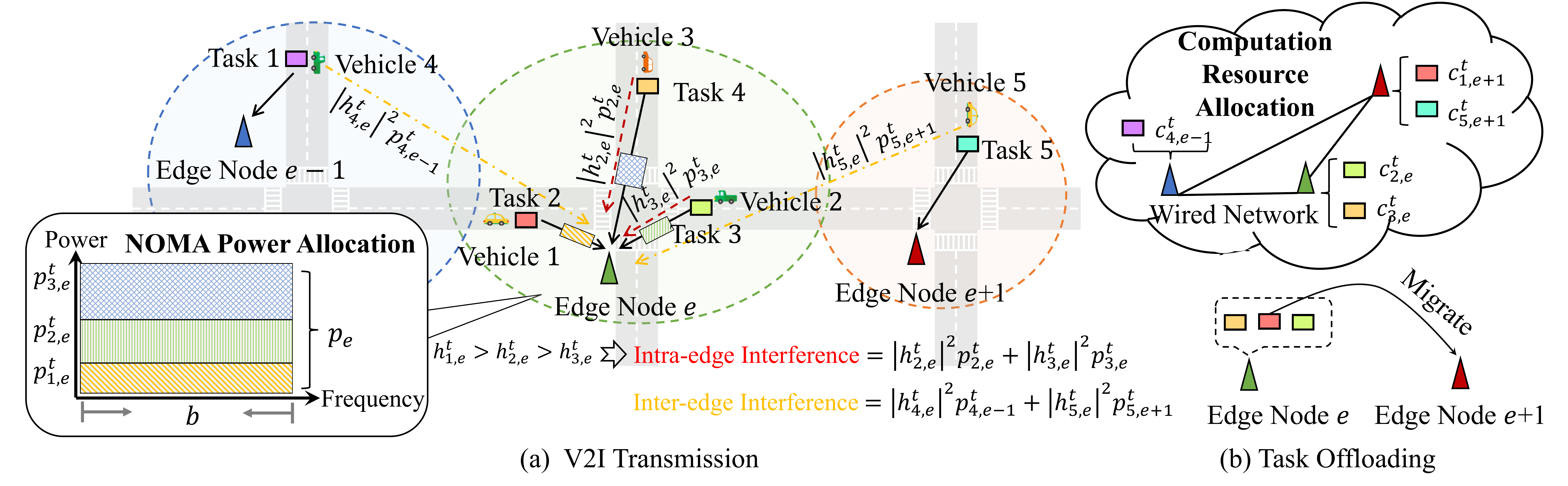}
  \caption{System model}
  \label{fig_2}
\end{figure}

\subsection{V2I Transmission Model}

The V2I transmission model is shown in Fig. \ref{fig_2}(a), and the intra-edge and inter-edge interferences are modeled based on the NOMA principle. 
We denote the transmission power of vehicle $v$ allocated by edge node $e$ at time $t$ as $p_{v, e}^{t}$.
The sum of allocated power cannot exceed the maximum power of V2I communications at edge node $e$, i.e., $\sum_{\forall v \in {V}_{e}^{t}} p_{v, e}^{t} \leq p_{e}$.
Then, the channel gain between the vehicle $v$ and edge node $e$ at time $t$ is denoted by $h_{v, e}^t$, which is computed by \cite{sun2020performance}:
\begin{equation}
	h_{v, e}^t = \frac{\eta_{v, e}}{{\operatorname{dist}_{v, e}^{t}}^{\varphi/2}}
\end{equation}
\noindent where $\eta_{v, e}$ is the Rayleigh distributed small scale fading, i.e., $\eta_{v, e} \sim \mathcal{CN}(0, 1)$, and $\varphi$ is the large scale path loss exponent.
Thus, the set of vehicles that have a worse instantaneous channel condition than vehicle $v$ is denoted by $\mathbf{V}_{h_{v, e}}^{t}$, which is represented by:
\begin{equation}
	\mathbf{V}_{h_{v, e}}^{t} = \left \{ v^{\prime} \mid  \left|h_{v^{\prime}, e}^t \right|^{2} < \left| h_{v, e}^t\right |^{2} , \forall v^{\prime} \in \mathbf{V}_{e}^{t} \right \}
\end{equation}

After determining the transmission power of each vehicle $v \in \mathbf{V}_{e}^{t}$, the observed signal by the edge node $e$ can be represented by \cite{islam2017power}:
\begin{equation}
	y_e^{t} = \sum_{\forall v \in \mathbf{V}_{e}^{t}} p_{v, e}^{t} s_{v, e}^{t} h_{v, e}^t + \sum\limits_{\forall e^{\prime} \in \mathbf{E} / \{e\}} \sum\limits_{\forall v^{\prime} \in \mathbf{V}_{e^{\prime}}^{t}} p_{v^{\prime}, e^{\prime}}^{t} s_{v^{\prime}, e^{\prime}}^{t} h_{v^{\prime}, e}^t + N_{0}
\end{equation}
where $s_{v, e}^{t}$ is the the message intended for vehicle $v$ and $N_{0}$ is the additive white Gaussian noise (AWGN).
According to the NOMA principle, the edge node $e$ can cancel the signals of vehicles with a better channel condition than vehicle $v$ via SIC.
Thus, the signal-to-interference-plus-noise ratio (SINR) between vehicle $v$ and edge node $e$ at time $t$ is denoted by $\mathrm{SINR}_{v, e}^t$, which can be computed by:
\begin{equation}
	\mathrm{SINR}_{v, e}^t = \frac{ |h_{v, e}^t| ^{2}  p_{v, e}^{t}}{ \underbrace{\sum\limits_{\forall v^{\prime} \in \mathbf{V}_{h_{v, e}}^{t}} |h_{v^{\prime}, e}^t|^2 p_{v^{\prime}, e}^{t}}_{\text {Intra-edge interference }} + \underbrace{\sum\limits_{\forall e^{\prime} \in \mathbf{E} / \{e\}} \sum\limits_{\forall v^{\prime} \in \mathbf{V}_{e^{\prime}}^{t}} |h_{v^{\prime}, e}^t|^2 p_{v^{\prime}, e^{\prime}}^{t}}_{\text {Inter-edge interference }} + N_{0}}
	\label{SINR}
\end{equation}
where $p_{v^{\prime}, e}^{t}$ is the transmission power of vehicle $v^{\prime} \in \mathbf{V}_{e}^{t}$ and $|h_{v^{\prime}, e}^t|^2$ is the channel coefficient of the interference link between vehicle $v^{\prime}$ and edge node $e$. 
The first and second parts in the denominator represent intra-edge and inter-edge interferences, respectively. 
Therefore, the uploading time of task $k_{v}^{t}$ requested by vehicle $v$ and transmitted to edge node $e$ is computed by:
\begin{equation}
	m_{v, e}^{t} = \frac{d_{k}}{b  \log _{2}\left(1+\mathrm{SINR}_{v, e}^t\right)}
	\label{equation_data_transmission_time}
\end{equation}
where $d_k$ is the data size of task $k_{v}^{t}$ and $b$ is the bandwidth of V2I communications, measured by Hz.

\subsection{Task Offloading Model}
The set of tasks uploaded by vehicles within the coverage of edge node $e$ at time $t$ is denoted by $\mathbf{K}_{e}^{t} = \{ k_{v}^{t}| \forall v \in \mathbf{V}_{e}^{t} \}$. 
As shown in Fig. \ref{fig_2}(b), each task $k_{v}^{t} \in \mathbf{K}_{e}^{t}$ can be either executed in edge node $e$ locally, or migrated to other edge nodes for processing.
The task offloading indicator is denoted by $q_{v, e}^{t}$, which indicates whether the task $k_{v}^{t}$ of vehicle $v$ is offloaded to the edge node $e$ at time $t$, and each task can only be offloaded in one edge node at least, i.e., $\sum_{\forall e \in \mathbf{E}} q_{v, e}^{t} = 1$.
Then, the set of tasks offloaded in the edge node $e$ can be represented by:
\begin{equation}
	\mathbf{K}_{q_e}^{t} = \left\{ k_{v}^{t} | q_{v, e}^{t} = 1, \forall v \in \mathbf{V}_{e^{\prime}}^{t}, \forall e^{\prime} \in \mathbf{E} \right\}
\end{equation}
which consists of local processing tasks uploaded by vehicles and tasks migrated from other edge nodes.
The computation resource (i.e., the CPU clock frequency) of task $k_{v}^{t} \in \mathbf{K}_{q_e}^{t}$ allocated by edge node $e$ is denoted by $c_{v, e}^{t}$.
The overall allocated computation resources cannot exceed the computation capability of edge node $e$, i.e., $ \sum_{\forall k_{v}^{t} \in {\mathbf{K}_{q_e}^{t} }} c_{v, e}^t \leq c_{e}$, where $c_e$ is the CPU clock frequency of edge node $e$.
Therefore, the execution time of task $a_{v}^{t}$ in edge node $e$ is denoted by $x_{v, e}^t$, which is computed by:
\begin{equation}
	x_{v, e}^t = \frac{ d_{k}  c_{k}}{c_{v, e}^t}
\end{equation}
where $d_{k}$ is the size of task $k_{v}^{t}$, and $c_{k}$ is the CPU cycles for processing one bit data of task $k_{v}^{t}$.

However, the task $k_{v}^{t}$ cannot be executed until the task data is received at the offloaded edge node $e$ when the task is requested by the vehicle $v$, which is without the radio coverage of edge node $e$.
Thus, we denote the wired transmission time of task $k_{v}^{t}$  transmitted by edge node $e$ and received at edge node $e^{\prime}$ by $w_{v, e}^{t}$, which is computed by:
\begin{equation}
	w_{v, e}^{t} = \begin{cases}
		0, &k_{v}^{t} \in \mathbf{K}_{e}^{t} \bigcap \mathbf{K}_{q_e}^{t}\\
		{d_{k}  \operatorname{dist}_{e, e^{\prime}}^{t}}  \zeta  / {z}, & k_{v}^{t} \in \mathbf{K}_{e}^{t} \bigcap \mathbf{K}_{q_{e^{\prime}}}^{t}\\
	\end{cases}
\end{equation}
\noindent where $\operatorname{dist}_{e, e^{\prime}} ^{t}$ is the distance between edge nodes $e$ and $e^{\prime}$, and $\zeta$ is a distance discount constant. 
The processing time of task $k_{v}^{t}$ in the edge node $e$ at time $t$ is denoted by $n_{v, e}^t$ and expressed by:
\begin{equation}
n_{v, e}^t= w_{v, e}^{t} + \sum_{\forall e^{\prime} \in \mathbf{E}} q_{v, e^{\prime}}^{t} x_{v, e^{\prime}}^t
\label{equation_execution_time}
\end{equation}
The processing time of task $k_{v}^{t}$ consists of the wired transmission time and execution time depending on the task offloading decisions.

\subsection{Cooperative Resource Optimization Problem}

The service time of task $k_v^t \in \mathbf{K}_{e}^{t}$ consists of the uploading time and processing time, which is denoted by $\psi_{v, e}^{t}$ and can be expressed by:
\begin{equation}
	\psi_{v, e}^{t} = m_{v, e}^{t} +  n_{v, e}^{t}
	\label{service time}
\end{equation}
The task $k_v^t$ is successfully serviced only if the service time is shorter than the deadline $t_k$.
Then, the service ratio of edge node $e$ is defined as the ratio between the number of successfully serviced tasks (i.e., be serviced before the deadline) and the number of requested tasks in the edge node $e$, which is denoted by $\Psi_{e}^{t}$ and represented by:
\begin{equation}
	\Psi_{e}^{t} = \frac{\sum_{\forall k_{v}^{t} \in \mathbf{K}_{e}^{t}} \mathbb{I} \left\{ \psi_{v, e}^{t} \leq t_{k} \right\} }{|\mathbf{K}_{e}^{t}|}
	\label{service ratio}
\end{equation}
\noindent where $|\mathbf{K}_{e}^{t}|$ is the number of tasks requested by vehicles within the coverage of edge node $e$, and $\mathbb{I} \left\{ \psi_{v, e}^{t} \leq t_{k} \right\}$ is a indicator function, and we have $\mathbb{I} \left\{ \psi_{v, e}^{t} \leq t_{k} \right\} =1$ if $\psi_{v, e}^{t} \leq t_{k}$, otherwise, $\mathbb{I} \left\{ \psi_{v, e}^{t} \leq t_{k} \right\} =0$.

Given a determined solution $(\mathbf{P}, \mathbf{Q}, \mathbf{C})$, where $\mathbf{P}$ denotes the determined V2I transmission power allocation, $\mathbf{Q}$ denotes the determined task offloading decisions, and $\mathbf{C}$ denotes the determined computation resource allocation, which is represented by:
\begin{equation}
\begin{cases}
\mathbf{P}&= \left \{ p_{v, e}^{t} \mid \forall v \in \mathbf{V}_{e}^t, \forall e \in \mathbf{E}, \forall t \in \mathbf{T}\right \}  \\
\mathbf{Q}&= \left \{ q_{v, e}^t \mid \forall v \in \mathbf{V}, \forall e \in \mathbf{E}, \forall t \in \mathbf{T} \right \} \\ 
\mathbf{C}&= \left \{ c_{v, e}^t \mid \forall v \in \mathbf{V}, \forall e \in \mathbf{E}, \forall t \in \mathbf{T} \right \}
\end{cases}
\end{equation}
This paper aims to maximize the sum of service ratios of edge nodes during the scheduling period by jointly optimizing task offloading decisions and heterogeneous resource allocation in NOMA-based VEC.
Thus, the cooperative resource optimization problem is formulated by:
\begin{equation}
	\begin{aligned}
		\operatorname{CRO}:&\max_{\mathbf{P}, \mathbf{Q}, \mathbf{C}} f_1= \sum_{\forall t \in \mathbf{T}} \sum_{ \forall e \in \mathbf{E}} \Psi_{e}^{t}  \\
		\text { s.t. }
        &\operatorname{C1}: \sum_{\forall v \in \mathbf{V}_{e}^{t}} p_{v, e}^{t} \leq p_{e}, \forall e \in \mathbf{E}, \forall t \in \mathbf{T}  \\
        &\operatorname{C2}: \sum_{\forall k_{v}^{t} \in {\mathbf{K}_{q_e}^{t} }} c_{v, e}^t \leq c_{e}, \forall e \in \mathbf{E}, \forall t \in \mathbf{T}\\
        &\operatorname{C3}: q_{v, e}^t \in \left \{0, 1\right \}, \forall v \in \mathbf{V}, \forall e \in \mathbf{E}, \forall t \in \mathbf{T}  \\
        &\operatorname{C4}: \sum_{\forall e \in \mathbf{E}} q_{v, e}^t = 1, \forall v \in \mathbf{V}, \forall t \in \mathbf{T} \\
      \end{aligned}
	\label{equ_objective}
\end{equation}

Constraint $\operatorname{C1}$ guarantees that the total transmission power allocated by the edge nodes cannot exceed the maximum power of V2I communications.
$\operatorname{C2}$ requires that the overall allocated computation resources cannot exceed the computation capacities of edge nodes.
Constraints $\operatorname{C3}$ and $\operatorname{C4}$ state that task offloading decision $q_{v, e}^t$ is an 0-1 integer variable, and each task can only be offloaded in one edge node at least.

\section{Proposed Solution}
As shown in Fig. \ref{fig_3}, the CRO is solved through decoupling the two subproblems, i.e., task offloading ($\mathcal{P}1$) and resource allocation ($\mathcal{P}2$). 
In particular, the $\mathcal{P}1$ is modeled as a non-cooperative game among edge nodes, which is proved as an EPG with the existence and convergence of NE.
To deal with $\mathcal{P}1$, we design the MAD4PG implemented at each edge node for task offloading to achieve the NE.
On the other hand, the $\mathcal{P}2$ is divided into two independent convex optimization problems.
We derive the optimal solution for heterogeneous resource allocation based on the gradient-based iterative method and KKT condition to handle $\mathcal{P}2$.
The interaction between the two solutions is described as follows.
First, the task offloading decisions are determined in advance based on the MAD4PG with the input of local system observation.
Then, the resource allocation is obtained via the optimal solution according to the task offloading decisions.
Further, in the NOMA-based VEC environment, the joint action of both task offloading and resource allocation is utilized to obtain the rewards for edge nodes via the designed potential function.
The procedure will be continued until the training for the MAD4PG is completed.

\begin{figure}
\centering
  \includegraphics[width=1\columnwidth]{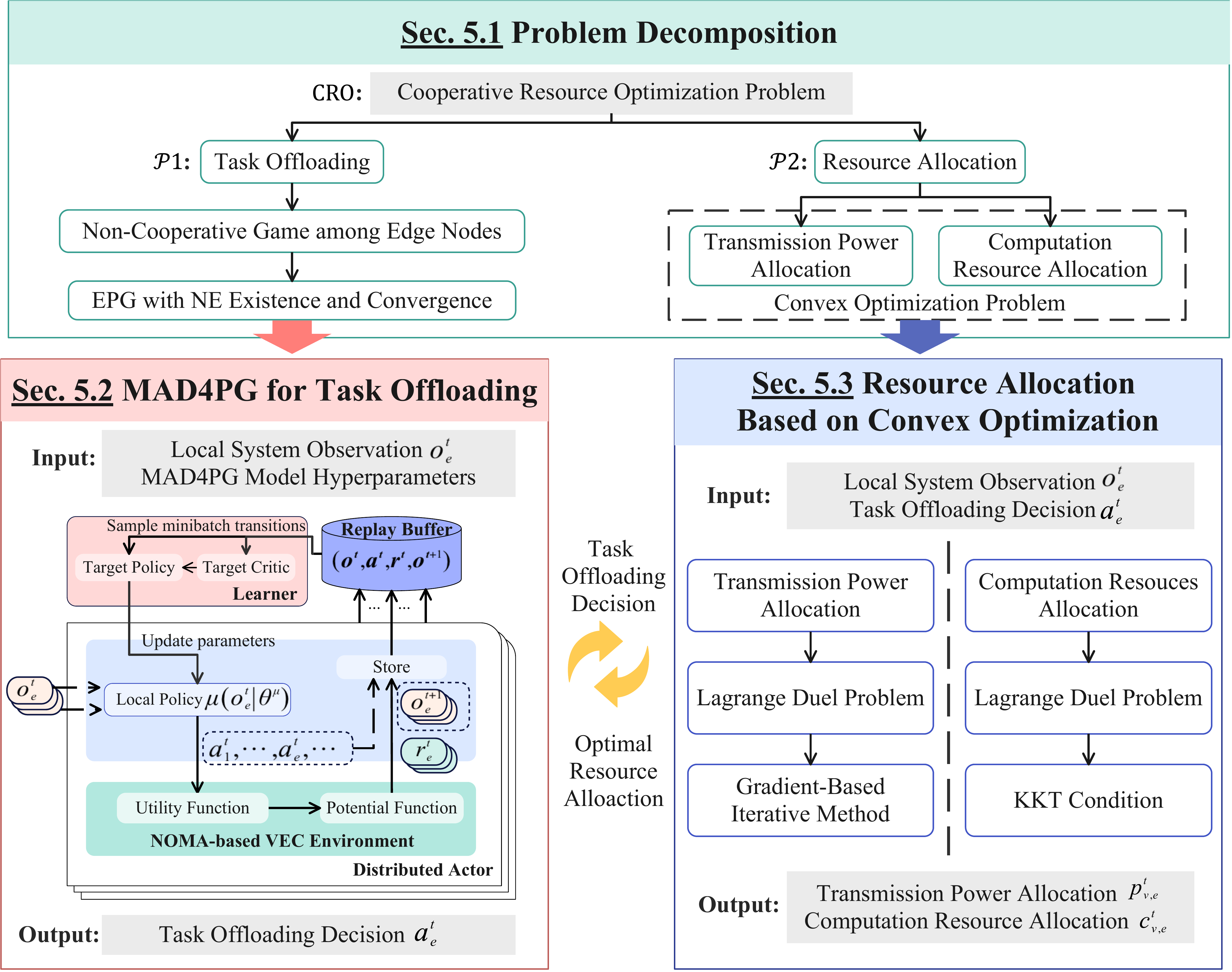}
  \caption{Overall of the proposed solution}
  \label{fig_3}
\end{figure} 

\subsection{Problem Decomposition}\label{Problem Decomposition}
In this section, we first decompose CRO into serval problems at each time slot.
Since the variables $\mathbf{P}^{t}$, $\mathbf{Q}^{t}$, and $ \mathbf{C}^{t}$ at time $t$ are independent of each other, and the four constraints are separable since the variables are not overlapped, the problem can be decomposed into two subproblems, which are formulated as follows.

\textit{1) Task Offloading:} the first subproblem $\mathcal{P}1$ with respect to $\mathbf{Q}^{t}$ concerns task offloading decisions of edge nodes, which is formulated by:
\begin{equation}
	\begin{aligned}
		\mathcal{P}1: &\max_{\mathbf{Q}^{t}} g_1= \sum_{ \forall e \in \mathbf{E}} \Psi_{e}^{t}  \\
		\text { s.t. }  
		&\operatorname{C5}: q_{v, e}^t \in \left \{0, 1\right \}, \forall v \in \mathbf{V}, \forall e \in \mathbf{E}  \\
        &\operatorname{C6}: \sum_{\forall e \in \mathbf{E}} q_{v, e}^t = 1, \forall v \in \mathbf{V} \\
	\end{aligned}
\end{equation}
Then, we model $\mathcal{P}1$ as a non-cooperative game among edge nodes, where edge nodes act as players to determine task offloading strategies independently.
The game model is represented as $\mathcal{G} = \left\{\mathbf{E}, \mathbb{S}, \left\{{U}_{e}\right\}_{\forall e \in \mathbf{E}} \right\}$, where $\mathbf{E}$ denotes the set of players;
$\mathbb{S}$ denotes the strategy space of the game, which is defined as the Cartesian products of all individual strategy sets of edge nodes, i.e., $\mathbb{S} = \mathbf{S}_{1} \times \ldots \times \mathbf{S}_{e} \times \ldots \times \mathbf{S}_{E}$, where $\mathbf{S}_{e}$ denotes the set of all possible strategies of edge node $e$.
Each element $\mathcal{S} \in \mathbb{S}$ is a strategy profile, and $\mathcal{S} = \left(\mathcal{S}_{1}, \ldots, \mathcal{S}_{e}, \ldots, \mathcal{S}_{E} \right)$, which can be rewritten into $\mathcal{S}=\left( \mathcal{S}_{e}, \mathcal{S}_{-e}\right)$, where $\mathcal{S}_{-e}$ denotes the joint strategy adopted by the opponents of edge node $e$, i.e., $\forall e^{\prime} \in \mathbf{E} \setminus \{e\}$. 
And $\mathcal{S}_{e}$ is the strategy of edge node $e$, which can be expressed by $\mathcal{S}_{e} = \left\{ q_{v, e}^t \mid \forall e \in \mathbf{E}, \forall v \in \mathbf{V}_{e}^{t} \right\}$;
${U}_{e}\left(\mathcal{S}\right)$ denotes the utility function of edge node $e$, which is defined as follows.
\begin{definition}
	The utility function of edge node $e$ denoted by ${U}_{e}\left(\mathcal{S}\right): \mathbb{S} \mapsto \mathbb{R}$ is defined as the sum of achieved service ratio of edge nodes under the strategy profile $\mathcal{S}$, where $\mathbb{R}$ is the set of real numbers.
	\begin{equation}
		{U}_{e}\left(\mathcal{S}\right) = \sum_{\forall e \in \mathbf{E}} \Psi_{e}^{t}
		\label{definition of Ue}
	\end{equation}
\end{definition}

Further, we prove that the non-cooperative game model $\mathcal{G}$ is an EPG with the existence and convergence of NE by giving a potential function as Eq. \ref{potential function}.
\begin{theorem}
Given a potential function of edge node $e$ as
\begin{equation}
	{F}_{e}\left(\mathcal{S}\right) = {U}_{e}\left(\mathcal{S}_{e}, \mathcal{S}_{-e}\right) - {U}_{e}\left(-\mathcal{S}_{e}, \mathcal{S}_{-e}\right)
	\label{potential function}
\end{equation}
the game $\mathcal{G}$ is an exact potential game.
\label{theorem_potential_game}
\end{theorem}
\noindent where ${U}_{e}\left(-\mathcal{S}_{e}, \mathcal{S}_{-e}\right)$ is the achievable utility when the strategy of edge node $e$ is not valid.
\begin{proof} See Appendix \ref{appendices_1}.
\end{proof}

\noindent In the game model $\mathcal{G}$, edge nodes attempt to achieve the NE \cite{chew2016potential} by maximizing their utility with conflicting interests.
\begin{definition}
	The strategy profile $\mathcal{S}^{*} \in \mathbb{S}$ is a pure-strategy Nash equilibrium \cite{chew2016potential} if and only if:
	\begin{equation}
		U_{e}\left(\mathcal{S}_{e}^{*}, \mathcal{S}_{-e}^{*}\right) \geq U_{e}\left(\mathcal{S}_{e}, \mathcal{S}_{-e}^{*}\right), \quad \forall \mathcal{S}_{e} \in \mathbf{S}_{e}, \forall e \in \mathbf{E}
	\end{equation}
\end{definition}
\begin{lemma}
	Given a potential function $F_{e}(\mathcal{S})$ as Eq. \ref{potential function}, the set of NE of the game $\mathcal{G}$ coincides with the set of NE for the game $\mathcal{G}^{F}=\left\{\mathbf{E}, \mathbb{S}, \left\{{F}_{e}\right\}_{\forall e \in \mathbf{E}} \right\}$, i.e.,
	\begin{equation}
		\mathcal{NE}(\mathcal{G}) \equiv \mathcal{NE}\left(\mathcal{G}^{F}\right)
	\end{equation}
	where $\mathcal{NE}$ denotes the Nash equilibrium set of a game.
	\label{lemma_neset}
\end{lemma}
\begin{proof} See Appendix \ref{appendices_2}
\end{proof}
\noindent Finally, we prove the existence of NE of the game model $\mathcal{G}$ based on Lemma \ref{lemma_neset}.
\begin{theorem}
	Given a potential function $F_{e}(\mathcal{S})$ as Eq. \ref{potential function}, the game $\mathcal{G}$ has at least one pure-strategy NE.
	\label{Nash equilibrium existence} 
\end{theorem}
\begin{proof} See Appendix \ref{appendices_3}
\end{proof}
\noindent On the other hand, due to the limited strategy space $\mathbb{S}$, the NE can converge in a finite number of steps. 
We establish the $\epsilon$-improvement path and $\epsilon$-equilibrium \cite{chew2016potential}, which is a strategy profile approximately close to an actual NE, then prove the convergence of NE.
\begin{definition} 
	A path $\rho=\left(\mathcal{S}^{0}, \mathcal{S}^{1}, \mathcal{S}^{2}, \ldots\right)$ is an $\epsilon$-improvement path \cite{chew2016potential} if in each step $i$, the utility of edge node $e$ is improved with the value $\epsilon$, i.e., $U_{e}\left(\mathcal{S}^{i+1}\right) > U_{e}\left(\mathcal{S}^{i}\right) + \epsilon, \exists \epsilon \in \mathbb{R}_{+}, \forall i$.
\end{definition}
\begin{definition}
	The strategy profile $\mathcal{\hat{S}} \in \mathbb{S}$ is an $\epsilon$-equilibrium \cite{chew2016potential} if and only if $\exists \epsilon \in \mathbb{R}_{+}$, and:
	\begin{equation}
		U_{e}\left(\mathcal{\hat{S}}_{e}, \mathcal{\hat{S}}_{-e}\right) \geq U_{e}\left(\mathcal{S}_{e}, \mathcal{\hat{S}}_{-e}\right) - \epsilon, \quad \forall \mathcal{S}_{e} \in \mathbf{S}_{e}, \forall e \in \mathbf{E}
	\end{equation}
\end{definition}
\begin{theorem}
	For the game $\mathcal{G}$, every $\epsilon$-improvement path is finite, and its endpoint is an $\epsilon$-equilibrium, which is a refinement of the original NE.
	\label{Nash convergence}
\end{theorem}
\begin{proof} See Appendix \ref{appendices_4}
\end{proof}

\textit{2) Resource Allocation:} the second subproblem $\mathcal{P}2$ with respect to $\mathbf{P}^{t}$ concerns on transmission power allocation and $\mathbf{C}^{t}$ concerns on computation resource allocation, which is formulated as follows.
\begin{equation}
	\begin{aligned}
		\mathcal{P}2: &\min_{\mathbf{P}^{t}, \mathbf{C}^{t}} g_2= \sum_{ \forall e \in \mathbf{E}} \sum_{\forall k_{v}^{t} \in \mathbf{K}_{e}^{t}} \left( m_{v, e}^{t} +  n_{v, e}^{t} \right)  \\
		\text { s.t. }
        &\operatorname{C7}: \sum_{\forall v \in \mathbf{V}_{e}^{t}} p_{v, e}^{t} \leq p_{e}, \forall e \in \mathbf{E}  \\
        &\operatorname{C8}: \sum_{\forall k_{v}^{t} \in {\mathbf{K}_{q_e}^{t} }} c_{v, e}^t \leq c_{e}, \forall e \in \mathbf{E}\\
	\end{aligned}
\label{equ_p4}
\end{equation}

\noindent It is observed that the variables $\mathbf{P}^{t}$ and $\mathbf{C}^{t}$ in Eq. \ref{equ_p4} are independent of each other. 
The constraints $\operatorname{C7}$ and $\operatorname{C8}$ are separable since the variables are not overlapped.
Therefore, the subproblem $\mathcal{P}2$ can be divided into two independent problems, namely, transmission power allocation and computation resource allocation, which are formulated as follow.

\textbf{Transmission Power Allocation:} it is with respect to the variables $\mathbf{P}^{t}$ concerns on transmission power allocation, which is formulated as follows.
\begin{equation}
	\begin{aligned}
		\mathcal{P}3: &\min_{\mathbf{P}^{t}} g_3= \sum_{ \forall e \in \mathbf{E}} \sum_{\forall k_{v}^{t} \in \mathbf{K}_{e}^{t}}  \frac{d_{k}}{b  \log _{2}\left(1+\mathrm{SINR}_{v, e}^t\right)} \\
		&\text { s.t. } \operatorname{C7}: \sum_{\forall v \in \mathbf{V}_{e}^{t}} p_{v, e}^{t} \leq p_{e}, \forall e \in \mathbf{E}\\
	\end{aligned}
\label{equ_p5}
\end{equation}
It is observed that the variables related to the edge nodes are independent.
Thus, the $\mathcal{P}3$ can be further divided into multiple simple problems, where each one is only related to an edge node $e$. 
\begin{equation}
	\begin{aligned}
		\mathcal{P}4: & \max_{\mathbf{P}_{e}^{t}}  g_3^e= \sum_{\forall k_{v}^{t} \in \mathbf{K}_{e}^{t}} {b  \log _{2}\left(1+\mathrm{SINR}_{v, e}^t\right)} \\
		&\text { s.t. } \operatorname{C9}: \sum_{\forall v \in \mathbf{V}_{e}^{t}} p_{v, e}^{t} \leq p_{e}  \\
	\end{aligned}
\label{equ_p6}
\end{equation}
However, the $\mathcal{P}4$ is nonconvex due to the intra-edge and inter-edge interferences.
Then, we apply some approximations to convert $\mathcal{P}4$ into a convex problem.
In particular, a lower bound of $g_3^e$ can be obtained and formulated by \cite{papandriopoulos2006low}:
\begin{equation}
	g_3^e \geq \overline{g_3^e} = \sum_{\forall k_{v}^{t} \in \mathbf{K}_{e}^{t}} {b \left( \xi_{v, e}^{t} \log _{2}\mathrm{SINR}_{v, e}^t + \omega_{v, e}^{t} \right) }
	\label{bar g_4}
\end{equation}
where $\xi_{v, e}^{t}$ and $\omega_{v, e}^{t}$ are fixed and given by:
\begin{equation}
	\begin{aligned}
		\xi_{v, e}^{t} &= \overline{\mathrm{SINR}}_{v, e}^t \bigg/ ( 1 + \overline{\mathrm{SINR}}_{v, e}^t ) \\
		\omega_{v, e}^{t} &= \log _{2} (1+ \overline{\mathrm{SINR}}_{v, e}^t) - \frac{\overline{\mathrm{SINR}}_{v, e}^t}{1 + \overline{\mathrm{SINR}}_{v, e}^t} \log _{2}\overline{\mathrm{SINR}}_{v, e}^t
	\end{aligned}
\end{equation}
The lower bound is tight if ${\mathrm{SINR}}_{v, e}^t =\overline{\mathrm{SINR}}_{v, e}^t$.
Thus, the $\mathcal{P}4$ can be re-expressed in the relaxation as:
\begin{equation}
		\begin{aligned}
		\mathcal{P}5: & \max_{\mathbf{P}_{e}^{t}}  \overline{g_3^e}= \sum_{\forall k_{v}^{t} \in \mathbf{K}_{e}^{t}} {b \left( \xi_{v, e}^{t} \log _{2}\mathrm{SINR}_{v, e}^t + \omega_{v, e}^{t} \right) } \\
		&\text { s.t. } \operatorname{C9}: \sum_{\forall v \in \mathbf{V}_{e}^{t}} p_{v, e}^{t} \leq p_{e}  \\
	\end{aligned}
\label{equ_p7}
\end{equation}
Nevertheless, the $\mathcal{P}5$ is still nonconvex because the objective is not concave in $\mathbf{P}_{e}^{t}$. 
Given a new variable $\widetilde{p_{v, e}^t} = \log _{2} {p}_{v, e}^t$, the $\mathcal{P}5$ can be transformed as follows.
\begin{equation}
		\begin{aligned}
		\mathcal{P}6: & \max_{\widetilde{\mathbf{{P}}_{e}^{t}}}  \widetilde{g_3^{e}}= \sum_{\forall k_{v}^{t} \in \mathbf{K}_{e}^{t}} {b ( \xi_{v, e}^{t} \log _{2}\mathrm{\widetilde{SINR}}_{v, e}^t + \omega_{v, e}^{t} ) } \\
		&\text { s.t. } \operatorname{C10}: \sum_{\forall v \in \mathbf{V}_{e}^{t}} 2^{\widetilde{p_{v, e}^t}} \leq p_{e}  \\
	\end{aligned}
\label{transmission power allocation convex problem}
\end{equation}
where $\log _{2}\mathrm{\widetilde{SINR}}_{v, e}^t$ is given by:
\begin{equation}
	\begin{aligned}
		\log _{2}\mathrm{\widetilde{SINR}}_{v, e}^t &= \widetilde{p_{v, e}^t} + \log _{2} |h_{v, e}^t| ^{2} - \log _{2} \left( \sum\limits_{\forall v^{\prime} \in \mathbf{V}_{h_{v, e}}^{t}} |h_{v^{\prime}, e}^t|^2 2^{\widetilde{p_{v^{\prime}, e}^{t}}} \right. \\
		&+ \left. \sum\limits_{\forall e^{\prime} \in \mathbf{E} / \{e\}} \sum\limits_{\forall v^{\prime} \in \mathbf{V}_{e^{\prime}}^{t}} |h_{v^{\prime}, e}^t|^2 2^{\widetilde{p_{v^{\prime}, e^{\prime}}^{t}}} + N_{0}\right)
	\end{aligned}
\end{equation}
\noindent Thus, the $\mathcal{P}6$ is a standard concave maximization problem as well as a convex optimization problem since each constraint is a sum of convex exponentials, and each term in the objective sum is concave.

\textbf{Computation Resource Allocation:} it is with respect to the variables $\mathbf{C}^{t}$ concerns computation resource allocation, which is formulated as follows.
\begin{equation}
	\begin{aligned}
		\mathcal{P}7: &\min_{\mathbf{C}^{t}} g_4= \sum_{ \forall e \in \mathbf{E}} \sum_{\forall k_{v}^{t} \in \mathbf{K}_{e}^{t}} ( w_{v, e}^{t} + \sum_{\forall e^{\prime} \in \mathbf{E}} q_{v, e^{\prime}}^{t} x_{v, e^{\prime}}^t) \\
		&\text { s.t. } \operatorname{C8}:  \sum_{\forall k_{v}^{t} \in {\mathbf{K}_{q_e}^{t} }} c_{v, e}^t \leq c_{e}, \forall e \in \mathbf{E}\\
	\end{aligned}
\label{equ_p10}
\end{equation}
Similar to the $\mathcal{P}3$ of Eq. \ref{equ_p5}, the $\mathcal{P}7$ can be further divided into multiple simple problems, where each is only related to an edge node $e$ and formulated as follows.
\begin{equation}
	\begin{aligned}
		\mathcal{P}8: &\min_{\mathbf{C}_{e}^{t}} g_4^e= \sum_{\forall a_{v}^{t} \in {\mathbf{K}_{q_e}^{t} }}   x_{v, e}^t \\
		&\text { s.t. } \operatorname{C11}:  \sum_{\forall k_{v}^{t} \in {\mathbf{K}_{q_e}^{t} }} {c_{v, e}^t} \leq c_{e}\\
	\end{aligned}
\label{computation resource allocation convex problem}
\end{equation}
\noindent where ${\mathbf{C}_e^t}$ represents the variables in ${\mathbf{C}^{t}}$ associated with edge node $e$.
Thus, the $\mathcal{P}8$ is a convex optimization problem, as the objective in Eq. \ref{computation resource allocation convex problem} is convex, and the constraint is linear.

\subsection{MAD4PG for Task Offloading} \label{MAD4PG for Task Offloading}
The MAD4PG model consists of several distributed actors, a learner, a NOMA-based VEC environment, and a replay buffer.
The primary components of MAD4PG are designed as follows.

\textit{1) System State:} The local observation of system state in the edge node $e$ at time $t$ is denoted by:
	\begin{equation}
		\boldsymbol{o}_{e}^{t}=\left\{e, t, \mathbf{Dist}_{\mathbf{V}_{e}^{t}}, \mathbf{D}_{\mathbf{K}_{e}^{t}}, \mathbf{C}_{\mathbf{K}_{e}^{t}}, \mathbf{T}_{\mathbf{K}_{v}^{t}}\right\}
	\end{equation} 
	\noindent where $e$ is the edge node index; 
	$t$ is the time slot index; 
	$\mathbf{Dist}_{\mathbf{V}_{e}^{t}}$ represents the set of distances between edge node $e$ and vehicle $v \in \mathbf{V}_{e}^{t}$ at time $t$,
	and $\mathbf{D}_{\mathbf{K}_{e}^{t}}$, $\mathbf{C}_{\mathbf{K}_{e}^{t}}$, and $\mathbf{T}_{\mathbf{K}_{v}^{t}}$ represent the set of data size, required computation resources, and deadline of task $k_{v}^{t} \in \mathbf{K}_{e}^{t}$ in edge node $e$ at time $t$, respectively.
	Thus, the system state at time $t$ can be denoted by $\boldsymbol{o}^{t}=\left\{\boldsymbol{o}_{1}^{t}, \ldots, \boldsymbol{o}_{e}^{t}, \ldots, \boldsymbol{o}_{E}^{t}\right\}$.

\textit{2) Action Space:} The action space of edge node $e$ consists of the offloading decision of tasks requested by vehicle $v \in \mathbf{V}_{e}^{t}$, which is denoted by:
	\begin{equation}
		\boldsymbol{a}_{e}^{t} = \left\{ q_{v, e^{\prime}}^t \mid \forall e^{\prime} \in \mathbf{E}, \forall v \in \mathbf{V}_{e}^{t} \right\}
	\end{equation}
	\noindent where $q_{v, e^{\prime}}^t \in \{0, 1\}$ indicates whether task $k_{v}^t$ is offloaded in the edge node $e^{\prime}$.
	The set of edge node actions is denoted by $\boldsymbol{a}^{t} = \left\{\boldsymbol{a}_{e}^{t}\mid \forall e \in \mathbf{E} \right\}$.
	
\textit{3) Reward Function:} In the game model, the objective of each edge node is to maximize its utilities.
	Therefore, the reward function of the system is defined as the achieved utilities of edge nodes at time $t$, which is represented by:
	\begin{equation}
		r\left(\boldsymbol{a}^{t} \mid \boldsymbol{o}^{t}\right)= {U}_{e}\left(\mathcal{S}_{e}, \mathcal{S}_{-e}\right) = \sum_{\forall e \in \mathbf{E}} \Psi_{e}^{t}
		\label{system reward}
	\end{equation}
	Further, the potential function of the game $\mathcal{G}$ is adopted as the reward of edge node $e$ with action $\boldsymbol{a}_{e}^{t}$ in the system state $\boldsymbol{o}^{t}$.
	\begin{equation}
		r_{e}^{t} = r\left(\boldsymbol{a}^{t} \mid \boldsymbol{o}^{t}\right)-r\left(\boldsymbol{a}_{-e}^{t} \mid \boldsymbol{o}^{t}\right)
		\label{edge_reward}
	\end{equation}
	\noindent where $r\left(\boldsymbol{a}_{-e}^{t} \mid \boldsymbol{o}^{t}\right)$ is the achieved system reward without the contribution of edge node $e$, and it can be obtained by setting null action set for edge node $e$.
	The set of rewards of edge nodes is denoted by $\boldsymbol{r}^{t} = \{r_{1}^{t}, \ldots, r_{e}^{t}, \ldots, r_{E}^{t}\}$.
	In the MAD4PG, the objective of each edge nodes $e \in \mathbf{E}$ is to maximize the expected return, which is represented by $R_{e}^{t} = \sum_{i \geq 0} \gamma^{i} r_{e}^{t+i}$, where $\gamma$ is the discount.

In the beginning of MAD4PG, the parameter of the local policy and critic networks are randomly initialized in the learner, which are denoted by $\theta^{\mu}$ and $\theta^{Q}$, respectively.
Then, the parameters of target policy and critic networks are initialized as the same as the corresponding local network, which are denoted by $\theta^{\mu^{\prime}}$ and $\theta^{Q^{\prime}}$, respectively.
\begin{equation}
	\begin{aligned}
		\theta^{\mu^{\prime}} \leftarrow \theta^{\mu}, \theta^{Q^{\prime}} \leftarrow \theta^{Q}\\
	\end{aligned}
\end{equation}
And the replay buffer $\mathcal{B}$ is initialized with a maximum size $|\mathcal{B}|$ to store replay experiences.
The procedure of MAD4PG is shown in Algorithm 1.

\begin{algorithm}[t]
	\caption{MAD4PG}
	Initialize the network weights;\\
	Initialize the replay buffer $\mathcal{B}$;\\
	Launch $J$ distributed actors and replicate network weights to the actors;\\
	\For{iteration $= 1$ to max-iteration-number}{
		\For{$t = 1$ to $T$}{
			\For{edge node $e=1$ to $E$}{
				Sample $M$ transitions of length $N$ from $\mathcal{B}$ randomly;\\
				Construct the target distributions;\\
				Compute the policy and critic network loss;\\
				Update the local policy and local critic networks;
			}
			\If{$t \mod t_{\operatorname{tgt}} = 0$}{
				Update the target networks;\\
			}
			\If{$t \mod t_{\operatorname{act}} = 0$}{
				Replicate network weights to the distributed actors;
			}
		}
	}
\end{algorithm}

On the other hand, there are $J$ distributed actors, which are launched to produce the replay experiences by interacting with the environment concurrently.
The parameters of the local policy network in the $j$-th actor are replicated from the local policy network of the learner, which are denoted by $\theta^{\mu}_{j}$. 
The initialized system state of each iteration is denoted by $\boldsymbol{o}^{0}$.
The task offloading action of edge node $e$ in the $j$-th actor at time $t$ is obtained based on the local observation of the system state:
\begin{equation}
	\boldsymbol{a}_{e}^{t}={\mu}\left(\boldsymbol{o}_{e}^{t} \mid \theta^{\mu}_{j}\right)+\epsilon  \mathcal{N}_{t}
\end{equation}
\noindent where $\mathcal{N}_{t}$ is an exploration noise to increase the diversity of edge actions, and $\epsilon$ is an exploration constant.
Then, the actions of edge nodes $\boldsymbol{a}^{t}$ are executed in the NOMA-based VEC environment, and the reward of each edge node can be obtained according to Eq. \ref{edge_reward}.
Finally, the interaction experiences including the system state $\boldsymbol{o}^{t}$, edge node actions $\boldsymbol{a}^{t}$, rewards of edge nodes $\boldsymbol{r}^{t}$, and next system state $\boldsymbol{o}^{t+1}$ are stored into the replay buffer $\mathcal{B}$.
The iteration will continue until the learner is finished.
The procedure of distributed actors is shown in Algorithm 2.

\begin{algorithm}[t]
	\caption{The $j$-th Distributed Actor}	
	\While{Learner is not finished}{
		Initialize a random process $\mathcal{N}$ for exploration;\\
		Receive the initial system state $\boldsymbol{o}_{1}$;\\
		\For{$t = 1$ to $T$}{
			\For{edge node $e=1$ to $E$}{
				Receive a local observation $\boldsymbol{o}_{e}^{t}$;\\
				Select a action $\boldsymbol{a}_{e}^{t}=\boldsymbol{\mu}\left(\boldsymbol{o}_{e}^{t} \mid \theta^{\mu}_{j}\right)+\mathcal{N}_{t}$;\\
			}
			Receive the reward $\boldsymbol{r}^{t}$ and the next system state $\boldsymbol{o}^{t+1}$;\\
			Store $\left(\boldsymbol{o}^{t}, \boldsymbol{a}^{t}, \boldsymbol{r}^{t}, \boldsymbol{o}^{t+1}\right)$ into replay buffer $\mathcal{B}$;
		}
	}
\end{algorithm}

A minibatch of $M$ transitions of length $N$ is sampled from replay buffer $\mathcal{B}$ to train the policy and critic networks of the learner. 
The transition of the $M$ minibatch is denoted by $\left(\boldsymbol{o}^{i:i+N}, \boldsymbol{a}^{i:i+N-1}, \boldsymbol{r}^{i:i+N-1}\right)$.
The target distribution of edge node $e$ is denoted by $Y_e^i$, which is computed by:
\begin{equation}
	Y_e^{i} = \sum_{n=0}^{N-1} \left( \gamma^{n} r_{e}^{i+n}\right)+\gamma^{N} Q^{\prime}\left(\boldsymbol{o}_{e}^{i+N}, \boldsymbol{a}^{i+N} \mid \theta^{Q^{\prime}} \right)
\end{equation}
\noindent where $\boldsymbol{a}^{i+N} = \{ \boldsymbol{a}_{1}^{i+N}, \ldots, \boldsymbol{a}_{e}^{i+N}, \ldots, \boldsymbol{a}_{E}^{i+N} \}$, and $\boldsymbol{a}_{e}^{i+N}$ is obtained via the target policy network, i.e., $\boldsymbol{a}_{e}^{i+N} = \mu^{\prime}(\boldsymbol{o}_{e}^{i+N} \mid \theta^{\mu^{\prime}})$.
The loss function of the critic network is represented by:
\begin{equation}
	{L}\left(\theta^{Q}\right)=\frac{1}{M} \sum_{i} \frac{1}{E} \sum_{e} \left(Y_e^{i}-Q\left(\boldsymbol{o}_{e}^{i}, \boldsymbol{a}^{i} \mid \theta^{Q}\right)\right)^{2}
\end{equation}
The parameters of the policy network are updated via policy gradient.
\begin{equation}
	\nabla_{\theta^{\mu}} \mathcal{J} = \frac{1}{M} \sum_{i} \frac{1}{E} \sum_{e} \nabla_{\boldsymbol{a}_{e}^{i}} Q\left(\boldsymbol{o}_{e}^{i}, \boldsymbol{a}^{i} \mid \theta^{Q}\right) \nabla_{\theta^{\mu}} \mu\left(\boldsymbol{o}_{e}^{i} \mid \theta^{\mu}\right)
\end{equation}
The parameters of the local policy network and local critic network are updated with the learning rate $\alpha$ and $\beta$.
Finally, the edge nodes update the parameters of target networks if $t \mod t_{\operatorname{tgt}} = 0$, where $t_{\operatorname{tgt}}$ is the target network parameter updating period.
\begin{equation}
	\begin{aligned}
			\theta^{\mu^{\prime}} &\leftarrow n \theta^{\mu}+(1-n)  \theta^{\mu^{\prime}}\\
			\theta^{Q^{\prime}} &\leftarrow n  \theta^{Q}+(1-n) \theta^{Q^{\prime}}
	\end{aligned}
\end{equation}
\noindent with $n \ll 1$.
The network parameter of policy network at $j$-th actor is also updated periodically, i.e., when $t \mod t_{\operatorname{act}} = 0$, where $t_{\operatorname{act}}$ is the network parameter updating period of the distributed actors.
\begin{equation}
	\theta_{j}^{\mu} \leftarrow \theta^{\mu^{\prime}}, \forall j
\end{equation}
where $\theta_{j}^{\mu}$ represents the parameters of the local policy network in the $j$-th distributed actor.
	
\subsection{Resource Allocation Based on Convex Optimization} \label{Resource Allocation Based on Convex Optimization}
\textit{1) Transmission Power Allocation :} To solve the convex optimization problem $\mathcal{P}6$, we first exploit the Lagrange dual method \cite{boyd2004convex} by introducing a Lagrange multiplier $\lambda_{e}^{t}$ into $\mathcal{P}6$.
Then, the Lagrange function is obtained by:
\begin{equation}
	\mathcal{L}(\widetilde{\mathbf{{P}}_{e}^{t}}, {\lambda}_{e}^{t} ) = \widetilde{g_3^{e}} -  {\lambda}_{e}^{t} (\sum_{\forall v \in \mathbf{V}_{e}^{t}} 2^{\widetilde{p_{v, e}^{t}}} - p_{e} )
	\label{Lagrange function}
\end{equation}
Further, the dual problem of $\mathcal{P}6$ is expressed as:
\begin{equation}
		\begin{aligned}
		\mathcal{P}9: & \min_{{\lambda}_{e}^{t}} \max_{\widetilde{\mathbf{{P}}_{e}^{t}}}  g_5 = \mathcal{L}(\widetilde{\mathbf{{P}}_{e}^{t}}, {\lambda}_{e}^{t} )\\
		&\text { s.t. } \operatorname{C12}: \mathbf{\lambda}_{e}^{t} \geq 0  \\
	\end{aligned}
\label{equ_p9}
\end{equation}
It is noted that the $\mathcal{P}9$ can be decomposed into a two-layer optimization problem.
The inner layer is represented as an optimization problem of $\widetilde{\mathbf{{P}}_{e}^{t}}$ with fixed ${\lambda}_{e}^{t}$, and the outer layer is represented as an optimization problem of ${\lambda}_{e}^{t}$ with fixed $\widetilde{\mathbf{{P}}_{e}^{t}}$.
In the outer layer, the dual variable ${\lambda}_{e}^{t}$ is iteratively updated through gradient descent.
\begin{equation}
	\mathbf{\lambda}_{e}^{t, (i+1)} = \max\{0, \mathbf{\lambda}_{e}^{t, (i)} + \sigma (\sum_{\forall v \in \mathbf{V}_{e}^{t}} 2^{\widetilde{p_{v, e}^{t}}} - p_{e} )\}
\end{equation}
where $\widetilde{p_{v, e}^{t}}$ is fixed; $\sigma$ is a sufficiently small constant, and $i$ is an iteration number.
Further, the inner dual maximization can be resolved by finding the stationary point of the Lagrange function in Eq. \ref{Lagrange function} with respect to $\widetilde{\mathbf{{P}}_{e}^{t}}$ and with fixed ${\lambda}_{e}^{t}$.
\begin{equation}
\frac{\partial \mathcal{L}\left(\widetilde{\mathbf{{P}}_{e}^{t}}, \mathbf{\lambda}_{e}^{t} \right)}{\partial \widetilde{p_{v, e}^{t}}}= b  \xi_{v, e}^{t}  - p_{v, e}^{t}(\lambda_{e}^{t} +\sum\limits_{\forall v^{\prime} \in \mathbf{V}_{h_{v, e}}^{t}} b  \xi_{v, e}^{t} |h_{v, e}^t|^2 \frac{\mathrm{SINR}_{v^{\prime}, e}^t}{|h_{v^{\prime}, e}^t| ^{2} p_{v^{\prime}, e}^{t}}) =0
\end{equation}
where the partial derivative is transformed back to the $\mathbf{P}_{e}^{t}$-space. 
Thus, the fixed-point equation can be formulated, and the transmission power of vehicle $v$ is updated by:
\begin{equation}
p_{v, e}^{t, {(i+1)}}=\frac{b \xi_{v, e}^{t}}{\lambda_{e}^{t,(i)}+\sum\limits_{\forall v^{\prime} \in \mathbf{V}_{h_{v, e}}^{t}} b  \xi_{v, e}^{t}|h_{v, e}^t|^2 {I}_{v^{\prime}, e}^{t, (i)} }
\end{equation}
where $\lambda_{e}^{t,(i)}$ and $p_{v, e}^{t, (i+1)}$ denote $\lambda_{e}^{t}$ in $i$-th iteration and $p_{v, e}^{t}$ in $(i+1)$-th iteration, respectively, and ${I}_{v^{\prime}, e}^{t, (i)}$ is given by:
\begin{equation}
	{I}_{v^{\prime}, e}^{t, (i)} = \sum\limits_{\forall v^{\prime} \in \mathbf{V}_{h_{v, e}}^{t}} |h_{v^{\prime}, e}^t|^2 p_{v^{\prime}, e}^{t, (i)} + \sum\limits_{\forall e^{\prime} \in \mathbf{E} / \{e\}} \sum\limits_{\forall v^{\prime} \in \mathbf{V}_{e^{\prime}}^{t}} |h_{v^{\prime}, e}^t|^2 p_{v^{\prime}, e^{\prime}}^{t, (i)} + N_{0}
\end{equation}
where $p_{v^{\prime}, e}^{t, (i)}$ and $p_{v^{\prime}, e^{\prime}}^{t, (i)}$ denote $p_{v^{\prime}, e}^{t}$ and $p_{v^{\prime}, e^{\prime}}^{t}$ in $i$-th iteration, respectively.

\textit{2) Computation Resource Allocation :} Similar to the transmission power allocation, we first introduce a Lagrange multiplier ${\lambda}_{e}^{t}$ into the $\mathcal{P}8$.
Then, the dual problem of $\mathcal{P}8$ can be expressed as:
\begin{equation}
	\begin{aligned}
		\mathcal{P}10: & \min_{\mathbf{\lambda}_{e}^{t}, \mathbf{{C}}_{e}^{t}}  g_6 = g_4^e - {\lambda}_{e}^{t} (\sum_{\forall k_{v}^{t} \in {\mathbf{K}_{q_e}^{t} }} {c_{v, e}^t} - c_{e} ) \\
		&\text { s.t. } \operatorname{C12}: \mathbf{\lambda}_{e}^{t} \geq 0  \\
	\end{aligned}
\label{equ_p9}
\end{equation}
Based on the KKT condition \cite{boyd2004convex}, we can get the following formulas:
\begin{equation}
	\begin{aligned}
		\nabla_{\mathbf{C}_e^{t}} g_4^e + \mathbf{\lambda}_{e}^{t}\nabla_{\mathbf{C}_e^{t}} (\sum_{\forall k_{v}^{t} \in {\mathbf{K}_{q_e}^{t} }} {c_{v, e}^t} - c_{e} ) &= 0,\\
		\mathbf{\lambda}_{e}^{t}( \sum_{\forall k_{v}^{t} \in {\mathbf{K}_{q_e}^{t} }} {c_{v, e}^t} - {c_{e}} ) &= 0,\\
		\mathbf{\lambda}_{e}^{t} &\geq 0
	\end{aligned}
\end{equation}
By solving the set of equations, the optimal solution of computation resource allocation for task $k_{v}^{t}$ can be obtained as follows.
\begin{equation}
	{c_{v, e}^{t}}^{\star} = \frac{1 / c_e \sqrt{d_k  c_k} } {\sum_{\forall k_{v}^{t} \in {\mathbf{K}_{q_e}^{t} }} 1 / c_e \sqrt{d_k  c_k}} , \forall k_{v}^{t} \in {\mathbf{K}_{q_e}^{t} } 
\end{equation}

\section{Performance Evaluation}

\subsection{Settings}

In this section, we implement a simulation model \footnote{The code can be found at https://github.com/neardws/Game-Theoretic-Deep-Reinforcement-Learning} by using Python 3.9.13 and TensorFlow 2.8.0 to evaluate the performance of the proposed solutions.
The simulation model is based on a Ubuntu 20.04 server with an AMD Ryzen 9 5950X 16-core processor (clocked at 3.4 GHz), two NVIDIA GeForce RTX 3090 graphic processing units, and 64 GB memory.
We consider the general scenario in a 3$\times$3 $\operatorname{km}^2$ square area, where $E = 9$ edge nodes such as 5G base stations and RSUs are uniformly distributed in the road map.
On the basis of referring to \cite{zhu2021decentralized}, \cite{liu2021rtds}, \cite{xu2021socially}, and \cite{zhou2019computation}, the simulation parameter settings are as follows. 
The computation capacities (i.e., the CPU clock frequencies) of edge nodes are different, which are set as uniformly distributed in $[3, 10]$ GHz \cite{zhou2019computation}.  
The communication range of V2I communications is set as $u_e =$ 500 m \cite{zhu2021decentralized}.

Further, the realistic vehicular trajectories are utilized as traffic inputs collected from Didi GAIA open data set \cite{didi} by extracting from a 3$\times$3 $\operatorname{km}^2$ area of Qingyang District, Chengdu, China, on 16 Nov. 2016.
In particular, we have examined three service scenarios with different periods.
Detailed statistics such as the total number of vehicle traces, the average dwell time (ADT) of vehicles, the variance of dwell time (VDT), the average number of vehicles (ANV) in each second, the variance of the number of vehicles (VNV), the average speed of vehicles (ASV), and the variance of speeds of vehicles (VSV) are summarized in Table \ref{table_traffic_characteristics}.
The heat maps of vehicle distribution within the scheduling period are shown in Fig. \ref{fig_4_heatmap} to exhibit the traffic characteristic under different scenarios better.
Comparing Figs. \ref{fig_4_heatmap}(a), \ref{fig_4_heatmap}(b), \ref{fig_4_heatmap}(c), and \ref{fig_4_heatmap}(d), it is noted that the vehicle density in the rush hour  (i.e., around 8:00, 13:00, and 18:00 on Nov. 16, 2016, Wed.) is much higher than that during the night (i.e., around 23:00) in the same area.
It can also be noted that the vehicle distributions are different during different scheduling periods.

\begin{table}[ht]\scriptsize
\centering
\caption{Traffic characteristics of each scenario}
\begin{tabular}[t]{ccccccccc}
\hline
\hline
Scenario&Time&Traces&ADT&VDT&ANV&VNV&ASV&VSV\\
\hline
No. 1&8:00-8:05&718&198.3(s)&123.8&474.6&11.6&5.22(m/s)&2.61\\
No. 2&13:00-13:05&862&188.5(s)&125.1&541.6&5.38&5.59(m/s)&2.73\\
No. 3&18:00-18:05&928&196.5(s)&122.5&608.0&7.76&4.60(m/s)&2.40\\
No. 4&23:00-23:05&359&173.7(s)&124.1&207.9&3.93&7.30(m/s)&3.16\\
\hline
\hline
\end{tabular}
\label{table_traffic_characteristics}
\end{table}

\begin{figure*}
\centering
  \includegraphics[width=1\columnwidth]{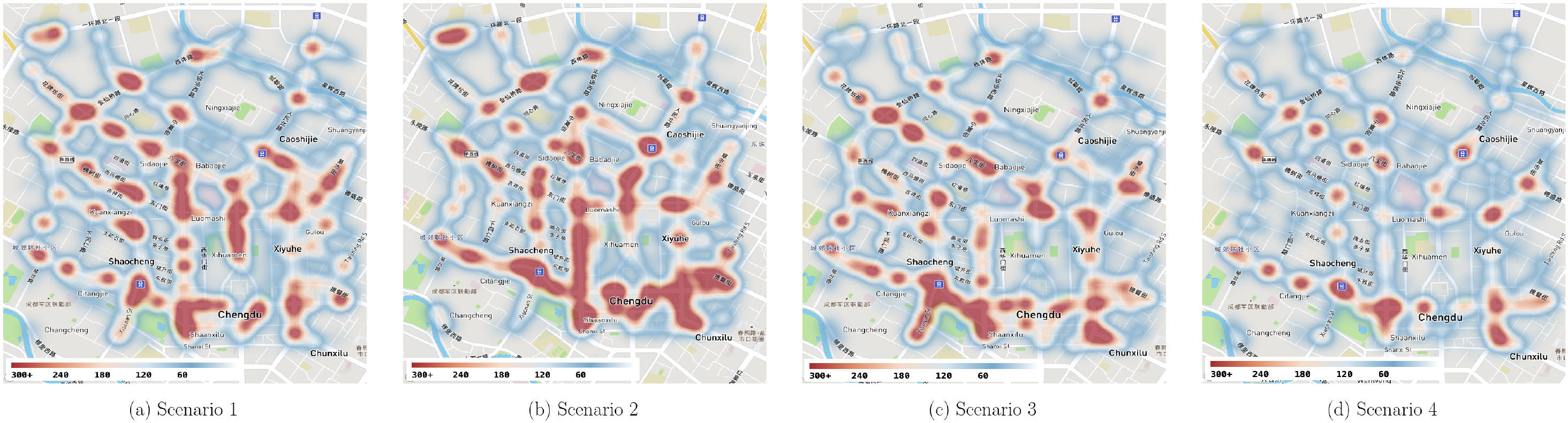}
  \caption{Heat map of the distribution of vehicles under different scenarios}
  \label{fig_4_heatmap}
\end{figure*} 

For the implementation of the MAD4PG, the architectures of the policy and critic networks are described as follows.
The local policy network is a five-layer fully connected neural network with three hidden layers, where the number of neurons is 256, 256, and 256, respectively.
The architecture of the target policy network is the same as the local policy network.
The local critic network is a five-layer fully connected neural network with three hidden layers, where the numbers of neurons are 512, 512, and 256, respectively.
The architecture of the target critic network is the same as the local critic network.
The Rectified Linear Unit (ReLU) is utilized as the activation function, and the Adam optimizer is used to update network weights.
The number of distributed actors is set as $J$=10. 
The primary system model parameters and algorithm parameters are shown in Table \ref{table_parameters}.

\begin{table*}[ht]\scriptsize
\centering
\caption{Parameters}
\begin{tabular}[t]{lll}
\hline
\hline
Parameters of System Model\\
\textbf{Parameter}&\textbf{Value}\\
\hline
Requested task size $d_k$ \cite{liu2021rtds}&[0.01, 5] MB\\
Computation cycles for processing 1-bit task data $c_k$ \cite{zhu2021decentralized} & 500 cycles/bit\\
Deadline of tasks $t_k$ \cite{liu2021rtds}&[5, 10] s\\
V2I communications bandwidth $b$ \cite{zhou2019computation}&20 MHz\\
Computation capability of edge node $c_e$ \cite{zhou2019computation}&[3, 10] GHz\\
Maximum power of V2I communications $p_e$ \cite{zhu2021decentralized}&1$\times 10^3$ mW\\
V2I communication range $u_e$ \cite{zhu2021decentralized}&500 m\\
Wired transmission rate $z$&50 Mbps\\
Distance discount $\zeta$& 6.667$\times 10^{-4}$\\
Additive white Gaussian noise $N_0$ \cite{xu2021socially}&-90 dBm\\
Large scale path loss exponent $\varphi$ \cite{xu2021socially}&3\\
\hline
\hline
Parameters of MAD4PG\\
\textbf{Parameter}&\textbf{Value}\\
\hline
Discount $\gamma$&0.996\\
Batch size $M$&256\\
Maximum replay buffer size $|\mathcal{B}|$&1$\times10^{6}$\\
Exploration constant $\epsilon$&0.3\\
Learning rate for policy network and critic network&1$\times10^{-4}$\\
Target network parameter updating period $t_{\operatorname{tgt}}$&100\\
Network parameter updating period of the distributed actors $t_{\operatorname{act}}$&1000\\
\hline
\hline
\end{tabular}
\label{table_parameters}
\end{table*}
 
For performance comparison, we implement four comparable algorithms as follows.
\begin{itemize}
	\item \textit{ORM}: it is divided into two stages: resource allocation and task offloading. Then, edge nodes prefer to migrate the tasks to other edge nodes.
	\item \textit{ORL}: it allocates the V2I transmission power and computation resource as ORM, and each edge node prefers executing the tasks locally.
	\item \textit{D4PG} \cite{barth2018distributed}: it jointly determines the task offloading decision, V2I transmission power allocation, and computation resource allocation by implementing a DDPG agent with the global system statues as inputs, where the utility function is adopted as the reward of the agent.
	\item \textit{MADDPG} \cite{zhang2021adaptive}: it allocates the V2I transmission power and computation resource as ORM, and implements the MADDPG in each edge nodes to determine task offloading decision independently, where the utility function is adopted as the rewards for edge nodes.
\end{itemize}

For performance evaluation, we collect the following statistics: the uploading time and processing time of each task; the total number of tasks executed locally, denoted by $K_{\operatorname{local}}$; the number of tasks migrated to other edge nodes, denoted by $K_{\operatorname{migrated}}$; the total number of tasks, denoted by $K_{\operatorname{total}}$, and the number of serviced tasks, denoted by $K_{\operatorname{serviced}}$. 
On this basis, four metrics named \textit{average processing time} (APT), \textit{average service time} (AST), \textit{average service ratio} (ASR), and \textit{cumulative reward} (CR) are obtained based on Eqs. \ref{equation_execution_time}, \ref{service time}, \ref{service ratio}, and \ref{system reward}, respectively. 
We further design the following two extra metrics for analysis.
\begin{itemize}
	\item \textit{Average Achieved Potential} (AAP): it is defined as the sum of edge rewards (i.e., the achieved potential) divided by the number of edge nodes during the scheduling period, which is computed by $ \frac{1}{E}\sum_{\forall e \in \mathbf{E}} \sum_{\forall t \in \mathbf{T}} r_{e}^{t}$.
	\item \textit{Proportion of Local Processing to Migration} (PLPM): The percentage of tasks that processed locally is computed by $P_{\operatorname{local}} = K_{\operatorname{local}}/K_{\operatorname{total}}$, whereas the percentage of tasks that migrated to other edge nodes is computed by $P_{\operatorname{migrated}} = K_{\operatorname{migrated}}/K_{\operatorname{total}}$, and we have $P_{\operatorname{local}} +P_{\operatorname{migrated}} = 1$. 
\end{itemize}

\subsection{Results and Analysis}

\begin{figure*}
\centering
  \includegraphics[width=1\columnwidth]{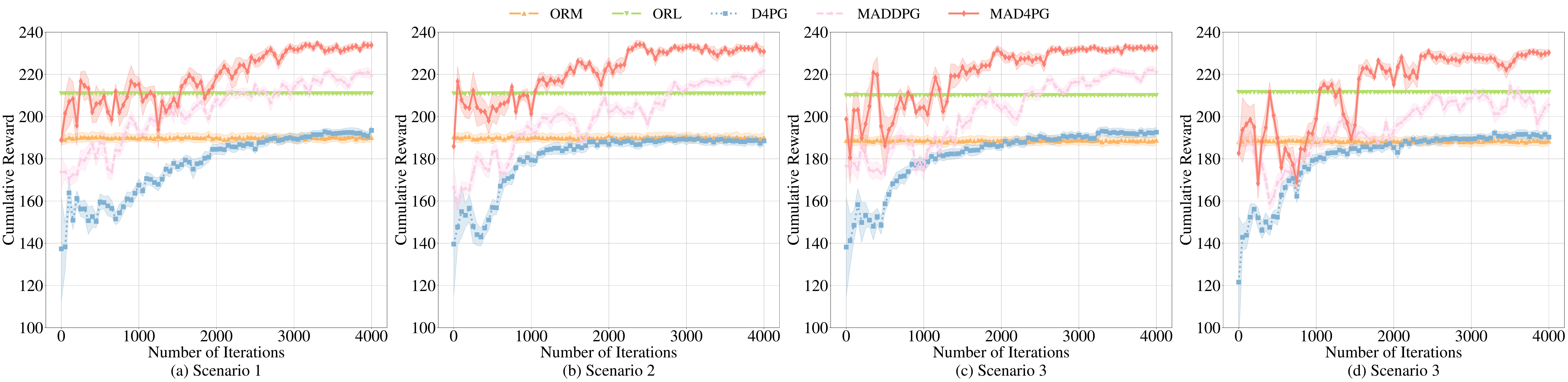}
  \caption{Algorithm convergence under different traffic scenarios}
  \label{fig Algorithm convergence}
\end{figure*} 

\textit{1) Algorithm Convergence:} Figure \ref{fig Algorithm convergence} compares the CR of the five algorithms in terms of convergence performance under different traffic scenarios. As noted, the convergence speed of  the proposed solution is just next to D4PG (i.e., around 3000 iterations), but it achieves the highest CR value (i.e., around 230). In contrast, D4PG and MADDPG converge in around 2000 and 3500 iterations and achieve the CR around 190 and 220, respectively. ORL and ORM achieve the CR of around 210 and 189, respectively. It is observed that the ORM, ORL, MADDPG, and MAD4PG can achieve the much higher CR than D4PG in the first 2000 iterations. The primary reason is that the proposed optimal resource allocation solution is used in ORM, ORL, MADDPG, and MAD4PG to make the performance better than D4PG, which jointly determines task offloading and resource allocation. On the other hand, due to the speed-up replay experience sampling by leveraging the distributed actors in MAD4PG, the proposed solution converges much faster than MADDPG, as well as achieves the highest CR under different traffic scenarios.

\begin{figure*}
\centering
  \includegraphics[width=1\columnwidth]{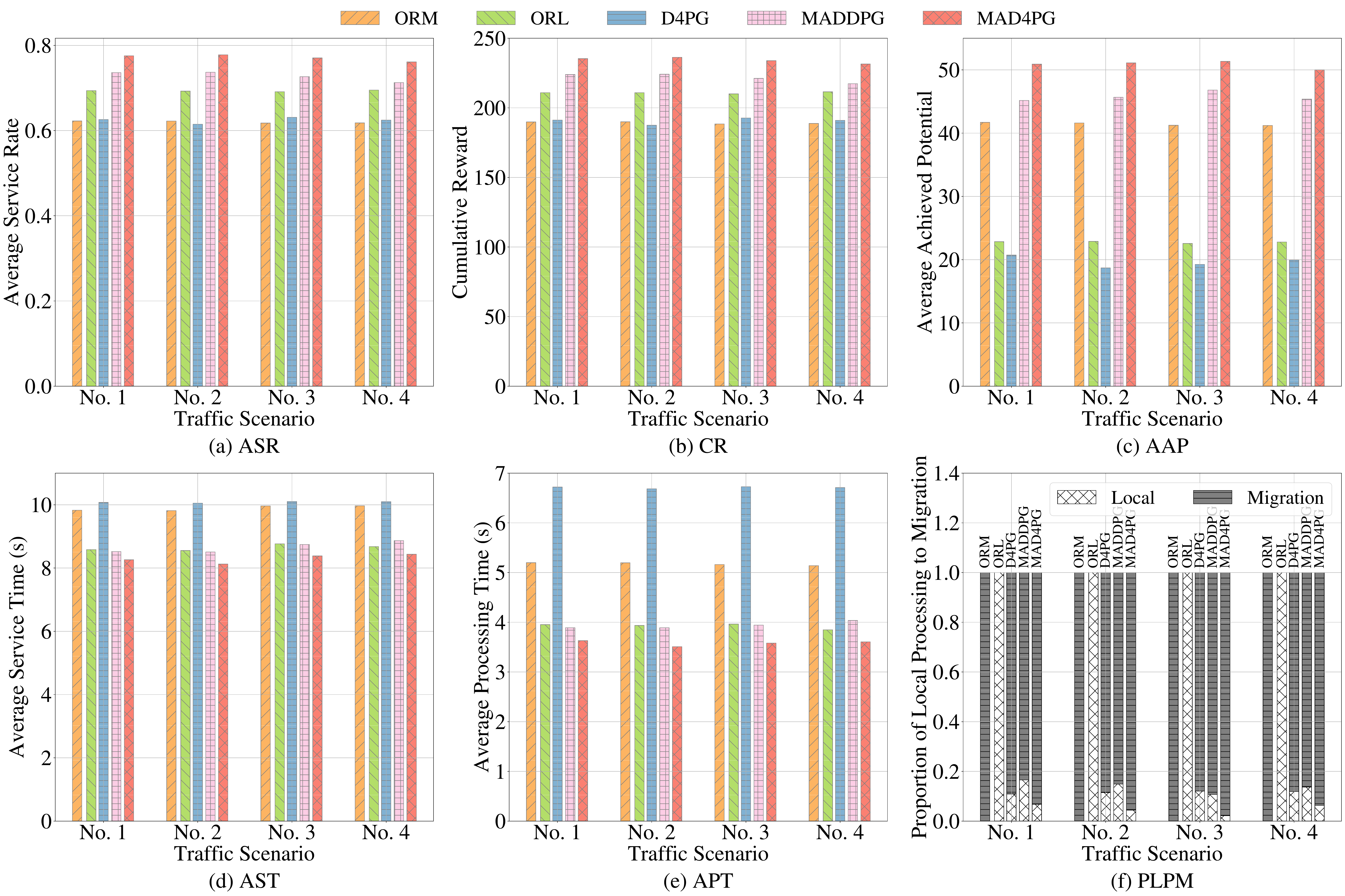}
  \caption{Performance comparison under different traffic scenarios}
  \label{different traffic scenarios}
\end{figure*}

\textit{2) Effect of Traffic Scenarios:} Figure \ref{different traffic scenarios} compares the five algorithms under different traffic scenarios. As demonstrated, Fig. \ref{different traffic scenarios}(a) compares the ASR of the five algorithms, and the MAD4PG achieves the highest ASR. Figure \ref{different traffic scenarios}(b) compares the CR of the five algorithms. As noted, the CR of the MAD4PG is higher than ORM, ORL, D4PG, and MADDPG. Figure \ref{different traffic scenarios}(c) shows the AAP of the five algorithms. It is expected that the MAD4PG achieves the highest AAP under all the scenarios, which indicates the advantage of the potential function as rewards of edge nodes in MAD4PG. Figures \ref{different traffic scenarios}(d) and \ref{different traffic scenarios}(e) compare the AST and APT of the five algorithms, respectively. It demonstrates that the MAD4PG can achieve cooperative communication and computation among edge nodes, improving the overall service ratio by minimizing the average service time of tasks. As expected, the APT of the MAD4PG is the lowest. It can be further justified by Fig. \ref{different traffic scenarios}(f), which shows that the tasks are more likely to migrate to other edge nodes for faster processing.

\begin{figure*}
\centering
  \includegraphics[width=1\columnwidth]{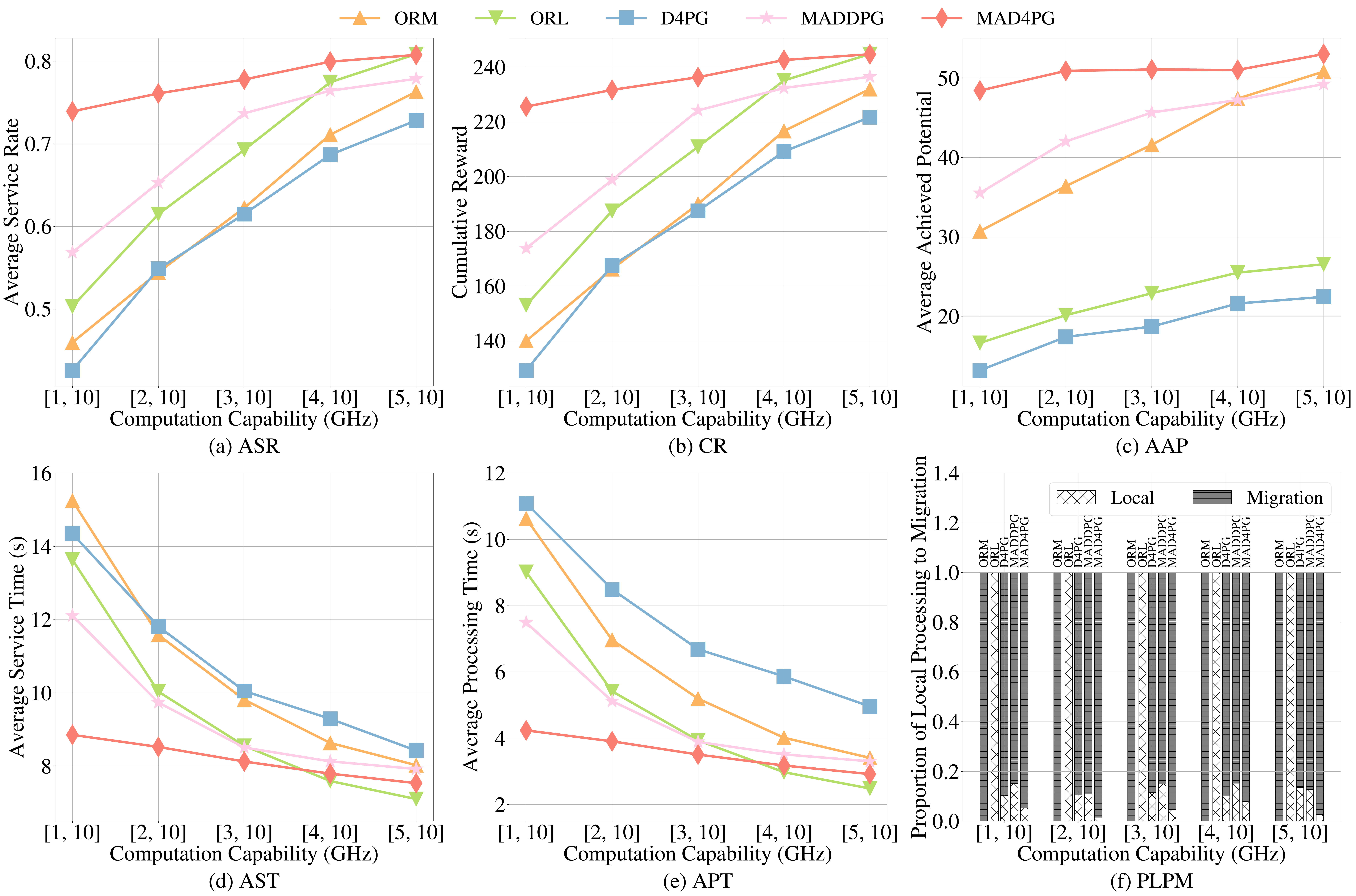}
  \caption{Performance comparison under different computation capabilities of edge nodes
  \label{different computation capability}}
\end{figure*} 

\textit{3) Effect of Computation Capability of Edge Nodes:} Figure \ref{different computation capability} compares the five algorithms under different computation capabilities of edge nodes. In this set of experiments, we consider that the computation capabilities of edge nodes follow the uniformed distributions, which increase from $c_e\sim[1, 10]$ GHz to $c_e\sim[5, 10]$ GHz. A more significant computation capability represents that more tasks can be executed. Figure \ref{different computation capability}(a) compares the ASR of the five algorithms. With the increasing computation capability, the ASR of all algorithms increases accordingly. Figure \ref{different computation capability}(b) compares the CR of the five algorithms. In particular, the MAD4PG achieves the highest CR.  Figure \ref{different computation capability}(c) compares the AAP of the five algorithms. As expected, the performance of all five algorithms gets better when the computation capability increases. Figure \ref{different computation capability}(d) compares the AST of the five algorithms. It is noted that the AST of ORL is lower than that of MAD4PG when the computation capabilities of edge nodes are more considerable (i.e., $c_e\sim[4,10]$ GHz and $c_e\sim[5, 10]$ GHz). The reason is that the gap between the computation capabilities of different edge nodes becomes smaller. Thus, the task processing time is shorter when the tasks are executed locally than offloaded to other edge nodes. It can be further verified in Fig. \ref{different computation capability}(e), which shows the APT of the five algorithms. It is noted that the APT of ORL is the shortest when the computation capability is larger; however, the ASR of ORL is smaller than MAD4PG. This is because the cooperation of communication and computation among edge nodes is more efficient in the MAD4PG. The advantage can be further convinced by Fig. \ref{different computation capability}(f), which shows the PLPM of the five algorithms.

\begin{figure*}
\centering
  \includegraphics[width=1\columnwidth]{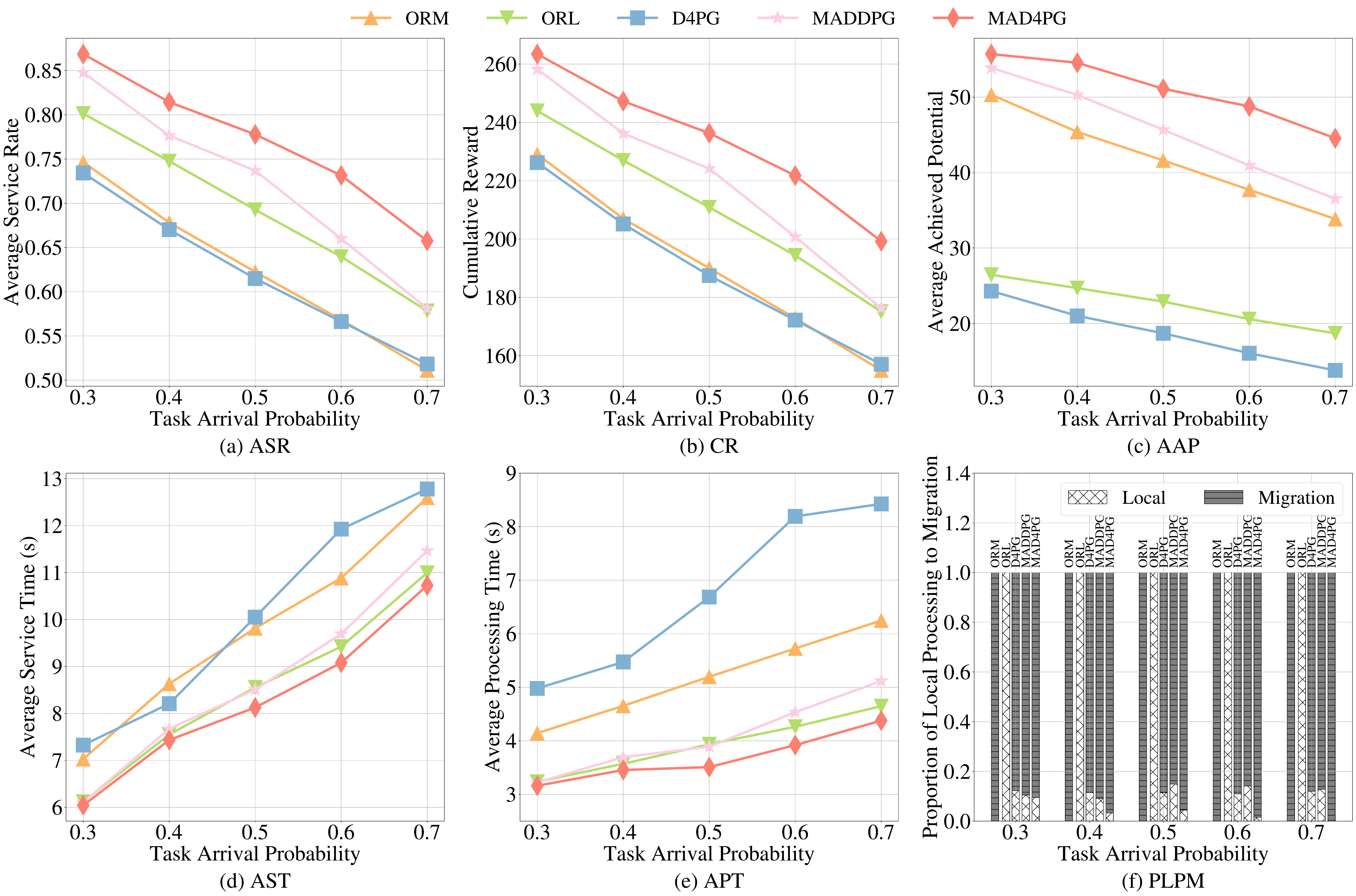}
  \caption{Performance comparison under different arrival probabilities of tasks
  \label{different task arrival probability}}
\end{figure*} 

\textit{4) Effect of Arrival Probability of Tasks:} Figure \ref{different task arrival probability} compare the five algorithms under different task arrival probabilities of vehicles. In this set of experiments, we consider that the task arrival probability of vehicles at each time slot increases from $\tau_{v}^{t}=0.3$ to $\tau_{v}^{t}=0.7$. As expected, the performance of all five algorithms gets worse when the task arrival probability increases. Figure \ref{different task arrival probability}(a) compares the ASR of the five algorithms, and the MAD4PG achieves the highest ASR. Figures \ref{different task arrival probability}(b) and \ref{different task arrival probability}(c) compare the CR and AAP of the five algorithms, showing that the MAD4PG can remain the highest CR and AAP across all cases, which indicates the advantages of MAD4PG by adopting the potential function as edge node reward. Figures \ref{different task arrival probability}(d) and \ref{different task arrival probability}(e) compare the AST and APT of the five algorithms. It is observed that the performance gap among the ORL, MADDPG, and MAD4PG is small when the task arrival probability increases from 0.3 to 0.4. The reason is that the scheduling effect is not significant when there are sufficient resources. Figure \ref{different task arrival probability}(f) compares the PLPM of the five algorithms. When the task arrival probability increases, the proportion of tasks processed locally in the MAD4PG decreases. The reason is that the tasks migrated to other edge nodes are more likely to be serviced before their deadlines.

\section{Conclusion and Future Work}

This paper presented a NOMA-based VEC architecture for cooperative communication and computation among edge nodes.
On this basis, the V2I transmission model was derived by considering the intra-edge and inter-edge interferences, and the task offloading model was derived by considering heterogeneous resources and cooperation among edge nodes.
Then, we formulated the CRO problem to maximize the service ratio.
Further, we decomposed CRO into two subproblems, namely, task offloading and resource allocation. 
The task offloading subproblem was modeled as an EPG with the existence and convergence of NE.
The MAD4PG algorithm was proposed, in which edge nodes act as agents with action space for determining the task offloading to achieve the NE.
In particular, the potential function of the game model was adopted as the rewards of edge nodes.
Then, the optimal solution was proposed to solve the remaining resource allocation problem based on the gradient-based iterative method and KKT condition.
Lastly, we built the simulation model with realistic vehicular trajectories extracted from different periods, and a comprehensive performance evaluation demonstrated the superiority of the proposed solutions.

In future work, we would like to further improve the system performance by considering the inner relationship between vehicle mobility and cooperative computation among edge nodes, e.g., the tasks can be migrated to the edge node where the vehicle will travel.
In addition, the end-edge-cloud hierarchical architecture for vehicular networks is expected to be incorporated to enhance performance by leveraging the cooperation among vehicles, edge nodes, and the cloud.
Finally, we would like to implement the solution model in the real world to verify system capability based on the actual vehicular network environments.

\begin{appendices}
\section{}

\subsection{Proof of Theorem \ref{theorem_potential_game}}
\label{appendices_1}
\begin{proof} According to Eq. \ref{potential function}, we have
\begin{equation}
	\begin{aligned}
		&{F}_{e}\left(\mathcal{{S}}^{\prime}_{e}, \mathcal{S}_{-e}\right) - {F}_{e}\left(\mathcal{S}_{e}, \mathcal{S}_{-e}\right) \\
		&={U}_{e}\left(\mathcal{S}^{\prime}_{e}, \mathcal{S}_{-e}\right) - {U}_{e}\left(-\mathcal{S}^{\prime}_{e}, \mathcal{S}_{-e}\right) - \left( {U}_{e}\left(\mathcal{S}_{e}, \mathcal{S}_{-e}\right) - {U}_{e}\left(-\mathcal{S}_{e}, \mathcal{S}_{-e}\right) \right)\\
		&={U}_{e}\left(\mathcal{S}^{\prime}_{e}, \mathcal{S}_{-e}\right) - {U}_{e}\left(\mathcal{S}_{e}, \mathcal{S}_{-e}\right) + {U}_{e}\left(-\mathcal{S}_{e}, \mathcal{S}_{-e}\right) - {U}_{e}\left(-\mathcal{S}^{\prime}_{e}, \mathcal{S}_{-e}\right)\\
		&={U}_{e}\left(\mathcal{S}^{\prime}_{e}, \mathcal{S}_{-e}\right) - {U}_{e}\left(\mathcal{S}_{e}, \mathcal{S}_{-e}\right)
	\end{aligned}
\end{equation}
Thus, the Theorem \ref{theorem_potential_game} is proved.
\end{proof}

\subsection{Proof of Lemma \ref{lemma_neset}}
\label{appendices_2}
\begin{proof}
Assume that $\mathcal{S}^{*}$ is a Nash equilibrium of game $\mathcal{G}$, we have
\begin{equation}
	U_{e}\left(\mathcal{S}_{e}^{*}, \mathcal{S}_{-e}^{*}\right) - U_{e}\left(\mathcal{S}_{e}, \mathcal{S}_{-e}^{*}\right) \geq 0, \quad \forall \mathcal{S}_{e} \in \mathbf{S}_{e}, \forall e \in \mathbf{E}
\end{equation}
According to the definition of the exact potential game, we have
\begin{equation}
	F_{e}\left(\mathcal{S}_{e}^{*}, \mathcal{S}_{-e}^{*}\right) - F_{e}\left(\mathcal{S}_{e}, \mathcal{S}_{-e}^{*}\right) \geq 0, \quad \forall \mathcal{S}_{e} \in \mathbf{S}_{e}, \forall e \in \mathbf{E}
\end{equation}
Therefore, $\mathcal{S}^{*}$ is also a Nash 
of game $\mathcal{G}^{F}$, and $\operatorname{NE}(\mathcal{G}^{F}) \subseteq \operatorname{NE}(\mathcal{G})$.
Similarly, $\operatorname{NE}(\mathcal{G}) \subseteq \operatorname{NE}(\mathcal{F})$.
\end{proof}

\subsection{Proof of Theorem \ref{Nash equilibrium existence}}
\label{appendices_3}
\begin{proof}
The strategy space $\mathbb{S}$ is closed and bounded. 
Hence, the potential function $F_{e}(\mathcal{S})$ has at least one maximum point in $\mathbb{S}$, which corresponds to the Nash equilibrium of the game $\mathcal{G}^{F}$. 
Then, according to Lemma. \ref{lemma_neset}, the game $\mathcal{G}$ has at least one pure-strategy Nash equilibrium.
\end{proof}

\subsection{Proof of Theorem \ref{Nash convergence}}
\label{appendices_4}
\begin{proof}
Since the strategy space $\mathbb{S}$ of game $\mathcal{G}$ is closed and bounded, $\exists F^{\max} \in \mathbb{R}$, and $F^{\max} < \infty$ such that $F^{\max} = \sup _{\mathcal{S} \in \mathbb{S}} F_{e}(S)$. 
Suppose that the path $\rho=\left(\mathcal{S}^{0}, \mathcal{S}^{1}, \ldots, \mathcal{S}^{i}, \ldots\right)$ is an $\epsilon$-improvement path, and it is infinite.
By the definition of $\epsilon$-improvement path, we have $U_{e}\left(\mathcal{S}^{i+1}\right) > U_{e}\left(\mathcal{S}^{i}\right) + \epsilon, \exists \epsilon \in \mathbb{R}_{+}, \forall i$.
Thus, we have $F_{e}\left(\mathcal{S}^{i+1}\right) > F_{e}\left(\mathcal{S}^{i}\right) + \epsilon^{\prime}, \exists \epsilon^{\prime} \in \mathbb{R}_{+}, \forall i$, where $\epsilon^{\prime}$ is a sufficiently small constant.
It can further imply that
\begin{equation}
	\begin{aligned}
		&F_{e}\left(\mathcal{S}^{i}\right) > F_{e}\left(\mathcal{S}^{0}\right) + i \cdot  \epsilon^{\prime}, \forall i \\
		&\lim _{i \rightarrow \infty} F\left(S^{i}\right) > \lim _{i \rightarrow \infty} \left \{ F_{e}\left(\mathcal{S}^{0}\right) + i \cdot  \epsilon^{\prime} \right\} =\infty
	\end{aligned}
\end{equation}
It contradicts $F^{\max} < \infty$, which indicates that the path $\rho$ must have finite steps and terminate at an $\epsilon$-equilibrium point.
\end{proof}
\end{appendices}

\section*{References}
\bibliography{mybibfile}

\authorbibliography[scale=0.3, wraplines=5, imagewidth=2cm, imagepos=l]{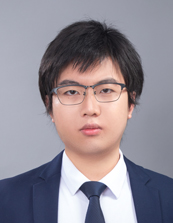}{}{\textbf{Xincao Xu} received the B.S. degree in network engineering from the North University of China, Taiyuan, China, in 2017. He is currently pursuing the Ph.D. degree in computer science at Chongqing University, Chongqing, China. His research interests include vehicular networks, edge computing, and deep reinforcement learning.\vspace{1\baselineskip}}

\authorbibliography[scale=0.1, wraplines=5, imagewidth=2cm, imagepos=l]{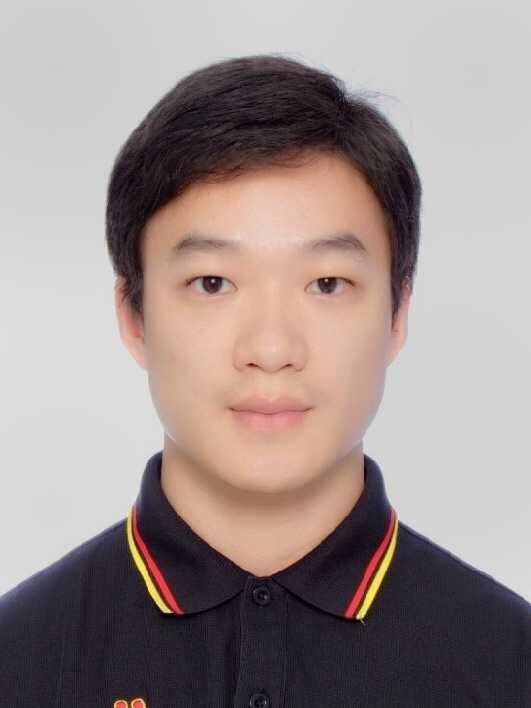}{}{\textbf{Kai Liu} received the Ph.D. degree in computer science from the City University of Hong Kong in 2011. He is currently a Full Professor with the College of Computer Science, Chongqing University, China. From 2010 to 2011, he was a Visiting Scholar with the Department of Computer Science, University of Virginia, Charlottesville, VA, USA. From 2011 to 2014, he was a Postdoctoral Fellow with Nanyang Technological University, Singapore, City University of Hong Kong, and Hong Kong Baptist University, Hong Kong. His research interests include mobile computing, pervasive computing, intelligent transportation systems, and the Internet of Vehicles.\vspace{1\baselineskip}}

\authorbibliography[scale=0.12, wraplines=5, imagewidth=2cm, imagepos=l]{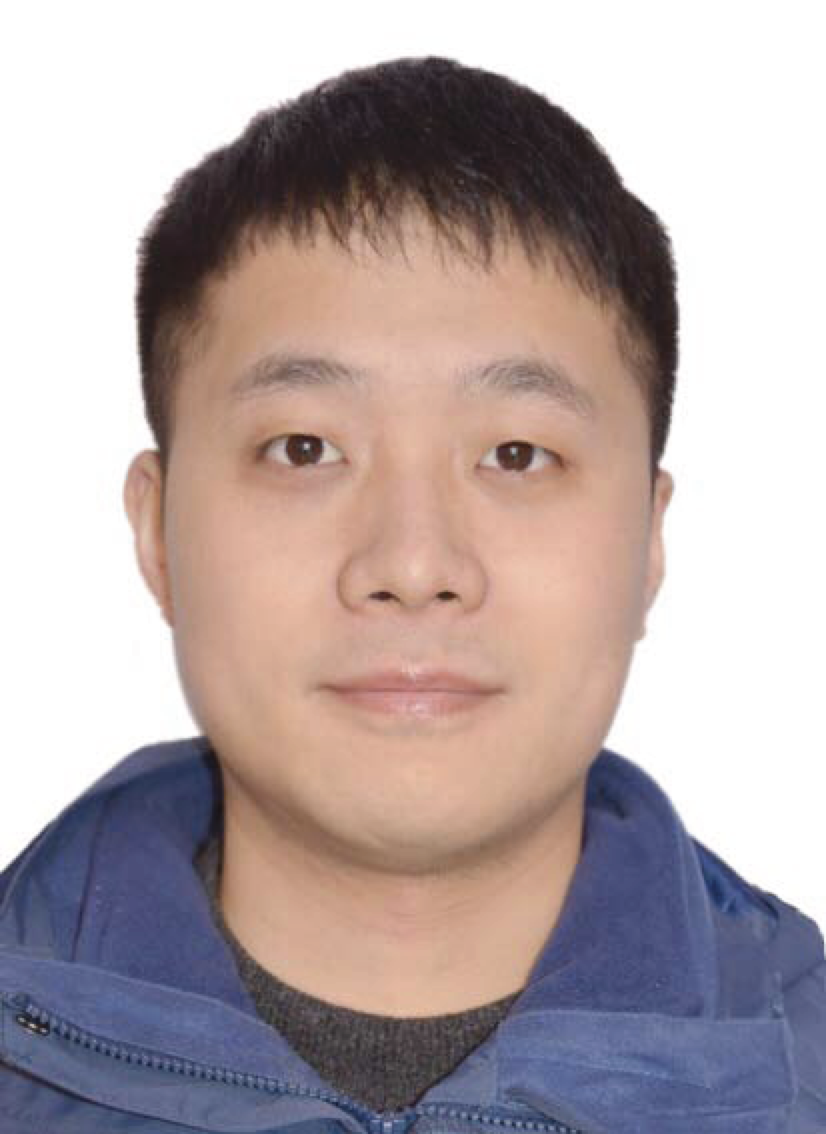}{}{\textbf{Penglin Dai} received the B.S. degree in mathematics and applied mathematics and the Ph.D. degree in computer science from Chongqing University, Chongqing, China, in 2012 and 2017, respectively. He is currently an Associate Professor with the School of Computing and Artificial Intelligence, Southwest Jiaotong University, Chengdu, China. His research interests include intelligent transportation systems and vehicular cyber-physical systems.\vspace{1\baselineskip}}

\authorbibliography[scale=0.18, wraplines=5, imagewidth=2cm, imagepos=l]{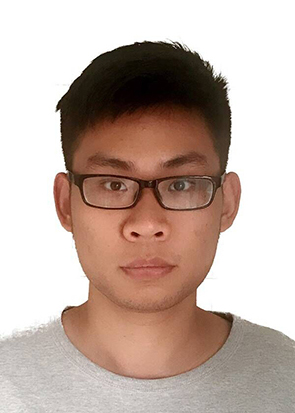}{}{\textbf{Feiyu Jin} received the M.S. Degree in computer science from Chongqing University, Chongqing, China, in 2020, where he is currently pursuing the Ph.D. degree. His research interests include pervasive computing, mobile computing, and intelligent transportation system.\vspace{1\baselineskip}}

\authorbibliography[scale=0.45, wraplines=5, imagewidth=2cm, imagepos=l]{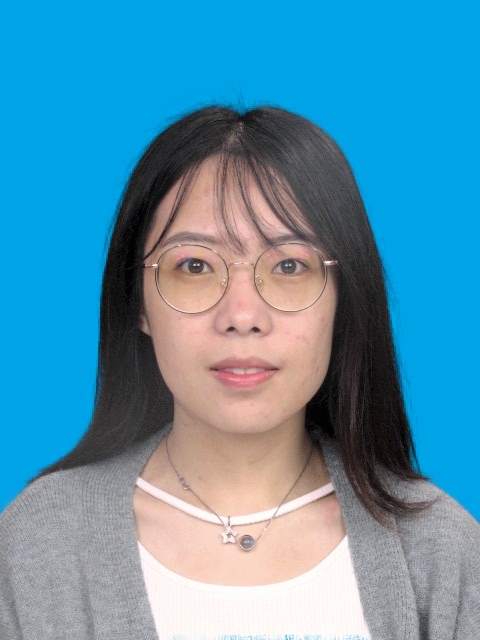}{}{\textbf{Hualing Ren} received the B.S. degree from the School of Software, Chongqing University, Chongqing, China, in 2018. She is currently pursuing the Ph.D. degree at the College of Computer Science, Chongqing University, China. Her research interests include intelligent transportation systems and mobile edge computing.\vspace{1\baselineskip}}

\authorbibliography[scale=0.5, wraplines=5, imagewidth=2cm, imagepos=l]{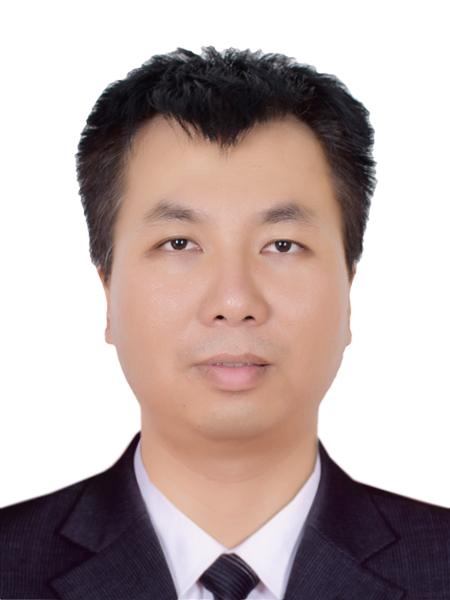}{}{\textbf{Choujun Zhan} received the BS degree in automatic control engineering from Sun Yat-Sen University, Guangzhou, China, in 2007 and the PhD degree in electronic engineering from City University of Hong Kong in 2012. After graduation, he worked as a postdoctoral fellow at the Hong Kong Polytechnic University. Now, he has been a Professor with the School of Computer, South China Normal University. His research interests include complex networks, time series modeling and prediction, epidemic spreading, information diffusion and machine learning.
\vspace{1\baselineskip}}

\authorbibliography[scale=0.10, wraplines=6, imagewidth=2cm, imagepos=l]{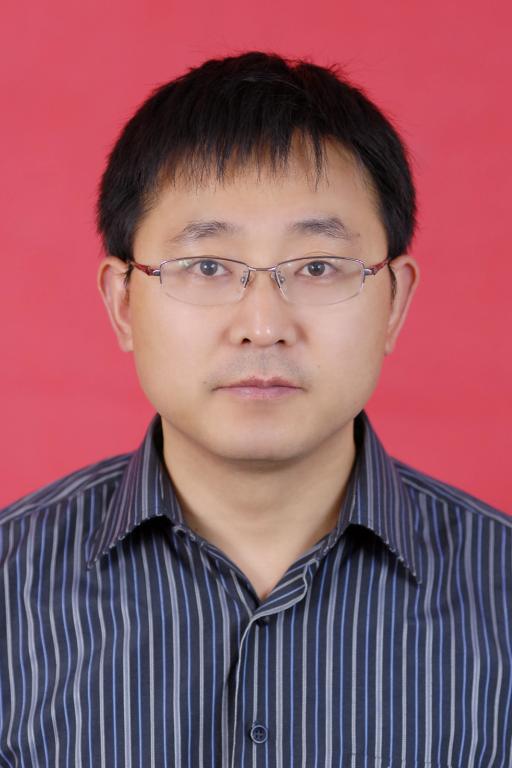}{}{\textbf{Songtao Guo} received the BS, MS, and PhD degrees in computer software and theory from Chongqing University, Chongqing, China, in 1999, 2003, and 2008, respectively. He is currently an associate dean and full professor with the College of Computer Science, Chongqing University, China. He was a senior research associate with the City University of Hong Kong from 2010 to 2011 and a visiting scholar with Stony Brook University, New York, from May 2011 to May 2012. He was a professor from 2012 to 2018 with Southwest University, China. His research interests include wireless sensor networks, wireless ad hoc networks, and parallel and distributed computing. He has published more than 70 scientific papers in leading refereed journals and conferences. He has received many research grants as a principal investigator from the National Science Foundation of China and Chongqing and the postdoctoral Science Foundation of China.
\vspace{0\baselineskip}}

\end{document}